\documentclass[11pt]{article}
\usepackage{amsthm,graphicx,amsmath,amssymb,amsfonts,framed}
\usepackage{slashed}
\usepackage{afterpage}
\usepackage{cite}
\usepackage{color}
\usepackage{hyperref}
\usepackage{color}

\definecolor{verde}{cmyk}{.83,.21,1,.08}

\newcommand{\sM}{\mathsf M}
\newcommand{\sN}{\mathsf N}

\newcommand{\cinf}{{C^\infty({\cal M})}}
\newcommand{\M}{\mathcal M}
\newcommand{\A}{\mathcal A}    
\newcommand{\HH}{\mathcal H}
\newcommand{\ds}{\slashed \partial}
\newcommand{\Ds}{\slashed D}
\newcommand{\sig}{{\boldsymbol{\sigma}}} 
\newcommand{\XX}{{\mathbb X}}

\font\mybb=msbm10 at 12pt
\def\bb#1{\hbox{\mybb#1}}

\newcommand{\Tr}[1]{\:{\rm Tr}\,#1}

\def\be{\begin{equation}}
\def\ee{\end{equation}}
\def\bea{\begin{eqnarray}}
\def\eea{\end{eqnarray}}

\newcommand{\del}{\partial}

\newcommand{\dd}{{\mathrm d}}

\newcommand{\spinind}[1]{\mathit{#1}}
\newcommand{\dotspinind}[1]{\mathit{\dot #1}}
\newcommand{\colorind}[1]{{\mathrm{#1}}}

\newcommand{\flavind}[1]{{\mathbf #1}}
\newcommand{\partind}[1]{{\sf #1}}

\newcommand{\sI}{\mathrm  I}
\newcommand{\sJ}{\mathrm J}
\newcommand{\sC}{\mathsf C}
\newcommand{\sD}{\mathsf D}

\newtheorem{prop}{Proposition}[section]
\newtheorem{thm}[prop]{Theorem}
\newtheorem{lemma}[prop]{Lemma}

\newtheorem{df}[prop]{Definition}

\begin{document}
\baselineskip=14pt
\begin{flushright}
%ICCUB-13-065
\end{flushright}
\begin{center}
{\Large{\bf Twisted spectral triple for the Standard Model\\[2pt] and\\[4pt]
spontaneous breaking of the Grand Symmetry}}
\bigskip

{\large{Agostino Devastato$^{ ab}$,\; Pierre Martinetti$^{ abcd}$}}
\bigskip

$^a$\,  Dipartimento di Fisica, Universit\`{a} di Napoli {\sl Federico
  II} \;\&\; $^b${ INFN, Sezione di Napoli}\\
{\it Monte S.~Angelo, Via Cintia, I-80126 Napoli}\\[5pt]
$^c$ Dipartimento di Matematica,  Universit\`{a} di Trieste\\
{\it Via Valerio
12/1 I-34127 Trieste}\\
$^d$ Dipartimento di Matematica,  Universit\`{a} di Genova\\
{\it Via Dodecaneso 5 I-16146 Genova}
\footnotetext{
agostino.devastato@na.infn.it, martinetti@dima.unige.it}
\end{center}
%\bigskip

\begin{abstract}
Grand symmetry models in noncommutative geometry, characterized by a
non-trivial action of functions on spinors, have been introduced 
to generate minimally (i.e. without adding new fermions)
and in agreement with the first order condition an
extra scalar field beyond the standard model, which both stabilizes the electroweak vacuum and
makes the computation of the mass of the Higgs
compatible with its experimental value. In this
paper, we use a twist in the sense of Connes-Moscovici to  cure
a technical problem due to the non-trivial action on spinors, that is the appearance
 together with the extra scalar field  of unbounded vectorial terms. The
twist makes these terms bounded and - thanks to a twisted version
of the first-order condition that we introduce here - also permits to
understand the breaking to the standard model as a dynamical process
induced by the spectral action, as conjectured in
\cite{Devastato:2013fk}. This is a spontaneous breaking from a pre-geometric Pati-Salam model to the almost-commutative
geometry of the standard model, with two Higgs-like fields: scalar and vector.
\end{abstract}

%\end{titlepage}

\section{Introduction}

Noncommutative geometry [NCG] provides a description of the standard model of elementary particles [SM] in
which the mass of the Higgs $-$ at unification scale $\Lambda$ $-$ is a function of the other parameters
of the theory, especially the Yukawa coupling of fermions \cite{Chamseddine:2007oz}. Assuming there is no new physics between
the electroweak and the unification scales (the ``big desert
hypothesis''),  the flow of this mass under
the renormalization group yields a prediction for the Higgs observable
mass $m_H$. It is well known that in the absence of new physics the three
constants of interaction fail to meet at a single unification
scale, but form a triangle which lays between $10^{13}$ and $10^{17}$~GeV. The situation can be improved by taking into account
higher order terms in the NCG action \cite{spectral-action-order6}, or
gravitational effects \cite{Agostino-grav}. Nevertheless, the
prediction of $m_H$ is not much sensible on the precise choice of the unification scale. Since the beginning of the model in the
early 90' \cite{Connes:1990ix, Connes:1996fu}, for $\Lambda$ between
$10^{13}$ and $10^{17}$GeV 
this prediction had
been around $170\text{ GeV}$, a value ruled out by Tevratron in
2008. Consequently, either the model should be abandoned, or the big desert hypothesis questioned.

The recent discovery of the Higgs boson with a mass $m_H\simeq 126
\,\text{Gev}$ suggests the big desert hypothesis should be
questioned. There is indeed an instability in the electroweak vacuum which is meta-stable rather than
stable (see \cite{near-critic} for the most recent update). There does
not seem to be a consensus in
the community whether this is an important problem or not: on the one
hand the
mean time of this meta-stable state is longer than the age of the
universe, on the other hand in some cosmological scenario the
meta-stabililty may be problematic \cite{cosmohiggs, Higgscosmoprocee}. Still, the fact that $m_H$ is almost
at the boundary value between the stable and meta-stable phases of the
electroweak vacuum suggests that ``something may be going on''. In
particular, particle physicists have shown how a
new scalar field  suitably coupled to the Higgs - usually
denoted $\sigma$ - can cure the
instability (e.g. \cite{C.-S.-Chen:2012vn,Elias-Miro:2012ys}).

Taking into account this extra field in the
NCG description of the SM induces a modification of the flow of the
Higgs mass, governed by the parameter $r=\frac{k_\nu}{k_t}$ which is the
ratio of the Dirac mass of the neutrino and of the Yukawa coupling of the quark top. Remarkably, for any value of 
$\Lambda$ between $10^{12}$ and $10^{17}$ Gev, there exists a
realistic value $r\simeq 1$  which brings back the computed value of
$m_H$ to $126 \text{
  Gev}$ \cite{Chamseddine:2012fk}.

The question is then to generate the extra field
$\sigma$ in agreement with the tools of noncommutative
geometry. Early attempts in this direction have been done in \cite{Stephan:2009fk}, but they require the
adjunction of new fermions (see \cite{Stephan:2013fk} for a
recent state of the art). In \cite{Chamseddine:2012fk}, a scalar $\sigma$ correctly coupled to the
Higgs is obtained without touching the fermionic content of the model,
simply by turning the Majorana mass $k_R$ of the neutrino into a field
\begin{equation}
  \label{eq:114}
  k_R \to k_R\, \sigma.
\end{equation}
Usually the bosonic fields in NCG are generated by inner
fluctuations of the geometry. However this does not work for the field
$\sigma$ because
of the first-order
condition
\begin{equation}
  \label{eq:121}
  [[D,a], Jb^*J^{-1}] = 0 \quad \forall a,b \in \A
\end{equation}
where $\A$ and $D$ are the algebra and the Dirac operator of the
spectral triple of the standard model, and $J$ the real structure.
 
In \cite{Chamseddine:2013fk,Chamseddine:2013uq} it was
shown how to obtain $\sigma$ by an inner fluctuation
that does not satisfy the first-order condition, but in such a way that
the latter is retrieved dynamically, as a minimum of the spectral
action. The field $\sigma$ is then interpreted as an excitation around this
minimum. Previously in \cite{Devastato:2013fk} another way had been
investigated  to generate 
$\sigma$ in agreement with the first-order condition, taking advantage of the fermion doubling in the Hilbert space $\cal H$ of the spectral
  triple of the SM \cite{fermionsdoublingNapoli2,
    fermiondoublingNaples1, fermiondoublingMarseille}. 

More specifically, under natural assumptions on the representation of
the algebra and an ad-hoc symplectic hypothesis, it is shown in
\cite{Chamseddine:2008uq} that the algebra in the spectral triple of
the SM should be a sub-algebra of $\cinf\otimes \A_F$, where
$\M$ is a Riemannian compact spin manifold (usually of dimension $4$) while 
\begin{equation}
\A_F=M_a(\bb H) \oplus M_{2a}(\bb C) \quad a\in \bb N.
 \end{equation}
The algebra of the standard model
\begin{equation}
  \label{eq:115}
  \A_{sm}:= \bb C \oplus \bb H \oplus M_3(\bb C)
\end{equation}
is obtained
from $\A_F$ for $a=2$, by the grading and the first-order conditions. 
Starting instead with the ``grand algebra'' ($a=4$)
\begin{equation}
  \label{eq:116}
  \A_G := M_4(\bb H) \oplus M_8(\bb C),
\end{equation}
one generates the field
$\sigma$ by a inner fluctuation which respects the first-order
condition imposed by the part $D_M$ of the Dirac
operator that contains the Majorana mass $k_R$ \cite{Devastato:2013fk}.  The breaking to $\A_{sm}$
is then obtained by the first-order
condition imposed by the free Dirac operator
\begin{equation}
\Ds:=\ds\otimes
{\mathbb I}_F
\label{eq:89}
\end{equation}
where ${\mathbb I}_F$ is the identity operator on the finite
dimensional Hilbert space ${\cal H}_F$ on which acts $\A_G$.

Unfortunately, before this breaking not
only is the first-order condition not satisfied, but the commutator
\begin{equation}
[\Ds, A] \quad A\in \cinf\otimes \A_G
\label{eq:117}
\end{equation}
is never bounded. This is problematic both for physics, because the
connection 1-form describing the gauge bosons is unbounded; and
from a mathematical point of view, because the construction of a
Fredholm module over $\A$ and Hochschild character cocycle depends on
the boundedness of the commutator \eqref{eq:117}.

In this paper, we solve this problem by using instead a
\emph{twisted spectral triple} $(\A, \HH, D,
\rho)$~\cite{Connes:1938fk}{\footnote{Also called
    $\sigma$-triple, but to avoid confusion with the field $\sigma$,
    we denote by $\rho$ the automorphism called $\sigma$ in \cite{Connes:1938fk}.}}. Rather than requiring the boundedness of the commutator,
one asks that there exists an automorphism $\rho$ of $\A$ such that the
twisted commutator 
\begin{equation}
  \label{eq:113}
 [D, a]_\rho :=  D a - \rho(a)D
\end{equation}
is bounded for any $a\in\A$. Accordingly, we introduce in Def. \ref{deftwistedfirstorder} 
a \emph{twisted first-order condition}
\begin{equation}
  \label{eq:122}
  [[D,a]_\rho, Jb^*J^{-1}]_\rho :=  [D,a]_\rho Jb^*J^{-1} -
  J\rho(b^*)J^{-1}[D,a]_\rho = 0 \quad \forall a,b\in \A.
\end{equation}
We then show that for a suitable choice of a subalgebra $\cal B$ of
$\A_G$, a \emph{twisted fluctuation}
of $\Ds +D_M$
that satisfies \eqref{eq:122} generates a field $\sig$ - slightly different
from the one of \cite{Chamseddine:2012fk} - together with an
additional  vector field $X_\mu$.

Furthermore, the breaking to the standard model is now spontaneous, as conjectured by Lizzi in
\cite{Devastato:2013fk}. Namely the reduction of the grand
algebra $\A_G$ to $\A_{sm}$ is obtained dynamically, as a minimum of the spectral
action. The scalar field $\sigma$ then plays a role similar as the one of the
Higgs in the electroweak symmetry breaking.

Mathematically, twists make sense as explained in
  \cite{Connes:1938fk},  for the
  Chern character of finitely summable spectral triples extends to the
  twisted case, and lands in ordinary (untwisted) cyclic
  cohomology. Twisted spectral triples have been introduced to deal
  with type III examples, such as those arising from transverse
  geometry of codimension one foliation, and have been used in
    various context since, like quantum statistical dynamical systems
    \cite{Greenfield:2014aa}. It is quite surprising that the same
    tool gives a possibility to implement in NCG the idea of a ``bigger symmetry beyond the SM''.
The main results of the paper are summarized in the following theorem.

\begin{thm}
\label{theo1}
Let $\cal H$ be the Hilbert space of the standard model described in
\S \ref{spectraltripleSM}.
There exists 
 a sub-algebra ${\cal B}$ of the grand algebra $\A_G$
  containing $\A_{sm}$ together with an automorphism $\rho$ of $\cinf\otimes{\cal B}$
such that
\smallskip 

i)  $T:= (\cinf\otimes{\cal B}, {\cal H}, \,\Ds + D_M;\,\rho)$ is a twisted
spectral triple satisfying the twisted $1^{\text{st}}$-order
condition~\eqref{eq:122};
\smallskip

ii) twisted
fluctuations of $\Ds + D_M$ by $\cinf \otimes\cal
B$ are parametrized by a scalar field $\boldsymbol \sigma$
and a vector field $X_\mu$; 
\smallskip

iii) the
spectral triple of the standard model is obtained from $T$ by
minimizing the potential of the vector field $X_\mu$ induced by the
spectral action coming from a twisted fluctuation of $\Ds$;

iv) the spectral triple of the standard model is also obtained by
minimizing the potential induced by the spectral action of a
twisted fluctuation of the whole Dirac operator $\Ds +D_M$. Such a
fluctuation provides a potential for the scalar field $\sig$, which is
minimum when $\Ds+D_M$ is fluctuated by $\cinf\otimes\A_{SM}$, that is when $\sig$
is the constant field $k_R$. 
\end{thm}

\noindent Explicitly, $\cal B$ is a sub-algebra $\bb H^2 \oplus \bb C^2
\oplus M_3(\bb C)$ of $\A_G$. Labeling the two copies of the
quaternions and complex algebras by the left/right spinorial indices $l,r$ and the left/right
internal indices L/R, that is 
\begin{equation}
  \label{eq:128}
  {\cal B} = {\bb H}^l_L \oplus {\bb H}^r_L \oplus {\bb C}^l_R \oplus {\bb
    C}^r_R \oplus M_3(\bb C),
\end{equation}
the automorphism $\rho$ is the exchange of the left/right spinorial indices:
\begin{equation}
  \label{eq:133}
  \rho\,(\,q^l_L,\, q^r_L,\, c^l_R,\, c^r_R,\, m) \to (q^r_L,\, q^l_L,\,
  c^r_R,\, c^l_R,\, m)
\end{equation}
where $m\in M_3({\bb C})$ while the $q$'s and $c$'s are quaternions and
complex numbers belonging to their
respective copy of $\bb H$ and $\bb C$.  
\bigskip

The paper is organized as follows. In section \ref{sectionSM} we
recall briefly the spectral triple of the standard model (\S
\ref{spectraltripleSM}), the tensorial notation used all
along the paper (\S\ref{mixing}), and the
results of \cite{Devastato:2013fk} on the grand algebra (\S \ref{se:tripleforsm}). We discuss the unboundedness
of the commutator \eqref{eq:117} in \S\ref{unboundsection}.  
Section \ref{section:twisting} deals with the twist. It begins with
the definition of the twisted first-order
condition in Def. \ref{deftwistedfirstorder}. In
\S \ref{subsec:representation} we fix the
representation of the grand algebra, which differs from the one used
in \cite{Devastato:2013fk}. It is used in \S\ref{subsec:twistedfirtsfree}
to build a  twisted
spectral triple with the free Dirac operator
(Prop. \ref{twisted-spec-triple}). In \S \ref{subsec:twistedfirstmaj} the twisted first-order condition for
$D_M$ yields the reduction to the algebra $\cal B$ and the
construction of the spectral triple $T$ (Prop. \ref{lemmasigma}). This
proves the first point of theorem \ref{theo1}.
In section \ref{section:twistedoperators} we compute the twisted fluctuations
$D_X$ of the free Dirac
operator $\Ds$ (\S\ref{potenziale}), and $D_\sigma$ of the
Majorana-Dirac operator $D_M$ (\S\ref{soussectionsigma}).
This yields the additional vector field in Prop. \ref{propdx}, and the
extra scalar field $\boldsymbol\sigma$ in
Prop. \ref{propsigma}, proving the second point of theorem \ref{theo1}.
In section \ref{section:breaking}, after some  generalities on the spectral
action in \S\ref{subsec:spectral}, we compute the generalized
Lichnerowicz formula for the twisted-fluctuated Dirac operator in
\S \ref{section:squaretwisteddirac}. The comparison with the
non-twisted case is made in \S~\ref{deviation}. The dynamical reduction of $\cal
B$ to the standard model by minimizing the potential of the additional
vector field is
obtained in \S \ref{secbreaking}. The potential of the scalar field is
treated in \S\ref{section:potentialsigma}, and the potential of
interaction between the vector and the scalar field in \S\ref{interaction}.
 These results are
discussed in section \ref{section:discussion}.   In  \S\ref{secdisc}
we stress how twisting the almost
commutative geometry of the SM may open the way to models where the
algebra is not the tensor product of a manifold by a finite
dimensional geometry. This justifies the choice of the representation
of $\A_G$ made in the present paper, but we show in \S\ref{subsec:inv}  that the results are the
same with the representation used in \cite{Devastato:2013fk}. 
\bigskip

\section{Standard model and the grand algebra}
\setcounter{equation}{0}
\label{sectionSM}

\subsection{The spectral triple of the standard model}
\label{spectraltripleSM}

The main tools of NCG \cite{Connes:1994kx} are encoded
  within a \emph{spectral triple} $(\A,
  \HH, D)$ where $\A$ is an involutive algebra acting on a Hilbert
  space $\HH$, and $D$ is a selfadjoint operator on $\HH$. These
  three elements come with two more operators, a real structure $J$
  \cite{Connes:1995kx} and a ${\mathbb Z}_2$-grading 
  $\Gamma$  that are generalizations to the noncommutative setting of 
  the charge conjugation and the chirality operators of quantum field
  theory. These five objects satisfy a set of properties guaranteeing that given any
  spectral triple with $\A$ unital and commutative, then there exists
  a closed Riemannian spin manifold $\M$ such that $\A=\cinf$ \cite{connesreconstruct}. These
  conditions still make sense in the noncommutative case \cite{Connes:1996fu}, hence the
  definition of a \emph{noncommutative geometry} as a spectral triple
  where the algebra $\A$ is non necessarily commutative. 

Among these conditions, the ones that play an important role in this work are the first-order condition \eqref{eq:121}, the boundedness and the grading conditions 
  \begin{equation}
[D,a]\in  {\cal B}(\HH),\quad  [\Gamma,a]=0 \quad \forall a\in\A,
\label{eq:136}
  \end{equation}
as well as the order-zero condition
\begin{equation}
  \label{eq:138}
  [a, Jb^* J^{-1}]=0 \quad \forall a,b\in \A.
\end{equation}

A gauge theory is
described by an \emph{almost commutative geometry}
\begin{equation}
  \label{eq:02}
  \A =\cinf\otimes \A_F,\; \HH = L^{2}(\M,S)\otimes \HH_F,\;  D=
  \slashed{\partial}\otimes {\mathbb I}_F + \gamma^5\otimes D_F,
\end{equation} which is the product of the canonical spectral triple $\left(C^{\infty}(\M),L^{2}(\M,S),\slashed{\partial}\right)$ associated to
a oriented closed spin manifold $\M$ of (even)
dimension $m$, by a finite
dimensional spectral triple
\begin{equation}
(\A_F, \HH_F, D_F).\label{eq:37}
\end{equation}
Here $L^{2}(\M,S)$ is the space of square integrable spinors on $\M$,
and{\footnote{Strictly speaking $\ds$ is not selfadjoint but
    essentially selfadjoint. We ignore this distinction and it is
    implicitly assumed all along the paper that we work with its closure.}}
\begin{equation}
\slashed{\partial} = -i\sum_{\mu=1}^m \gamma^\mu\nabla_\mu^S\quad \text{
  with } \quad\nabla_\mu^S =\del_\mu + \omega_\mu^S 
\label{eq:35}
\end{equation}
 is the
Dirac operator with $\gamma^\mu={\gamma^\mu}^\dagger$ the selfadjoint
Dirac matrices and $\omega_\mu^S$ the spin connection. The chirality operator $\gamma^5$ 
is a grading of $L^2(\M,S)$ which commutes with $\cinf$
and anticommutes with $\slashed\partial$.
The notation is justified assuming $\M$ has dimension $4$ (what we do
from now on, for simplicity and obvious physical reasons): $\gamma^5$
is then the product of the four Dirac matrices.\setcounter{footnote}{0}{\footnote{Most of the
results presented in this paper should work in arbitrary even
dimension: for the construction of real twisted spectral triples and their twisted
fluctuation, this has actually been shown in
\cite{Landi:2015aa}; for the spectral action, this still needs to be
checked. In odd dimension there may be some subtleties due to the
grading, that need to be further investigated.}} 

The choice of the finite dimensional spectral triple \eqref{eq:37} is dictated by
the physical contents of the theory. For the SM, the algebra is
$\A_{sm}$ given in \eqref{eq:115}, whose group of unitary elements yields the gauge group
of the standard model. The finite dimensional Hilbert space $\HH_F$ is spanned by the particle
content of the theory. The standard model has 96 such degrees of freedom:
8 fermions (electron, neutrino, up and down quarks with three
colors each) for N=3 generations and  two chiralities $L$, $R$, plus
antiparticles. Therefore one takes
\begin{equation}
\mathcal{H}_{F}=
\mathcal{H}_{R}\oplus\mathcal{H}_{L}\oplus\mathcal{H}_{R}^{c}\oplus\mathbb{\mathcal{H}}_{L}^{c}=\bb C^{96}.
\label{eq:56}
\end{equation}

The finite dimensional Dirac operator $D_F = D_0 + D_R$ is a
$96\times96$ matrix where
\begin{equation}
D_0:=\left(\begin{array}{cccc}
0_{8N} & \mathcal{M}_0 & 0_{8N} & 0_{8N}\\
\mathcal{M}_0^{\dagger} & 0_{8N} & 0_{8N} & 0_{8N}\\
 0_{8N}& 0_{8N} & 0_{8N} & \bar{\mathcal{M}}_0\\
0_{8N} &0_{8N} & \mathcal{M}_0^{T} & 0_{8N}
\end{array}\right)\;\text{ and }\; 
D_R:=\left(\begin{array}{cccc}
0_{8N} & 0_{8N} & \mathcal{M}_{R} & 0_{8N}\\
0_{8N} & 0_{8N} & 0_{8N} & 0_{8N}\\
\mathcal{M}_{R} ^\dagger& 0_{8N} & 0_{8N} &0_{8N}\\
0_{8N} &0_{8N} & 0_{8N}& 0_{8N}
\end{array}\right).
\label{eq:D_F Modello Standard-1-1}
\end{equation}
The matrix $\mathcal{M}_0$ contains the Yukawa couplings of 
fermions, the Dirac mass of neutrinos, the Cabibbo matrix and the
mixing matrix for neutrinos. The matrix $\mathcal{M}_{R}$
contains the Majorana mass of neutrinos. Explicitly
\begin{equation}
\mathcal{M}_0  =  \left(\begin{array}{cc}
M_{u} & 0_{4}\\
0_{4} & M_{d}
\end{array}\right)\otimes \bb I_N\quad\quad
\mathcal{M}_{R}  =  \left(\begin{array}{cc}
M_{R} & 0_{4}\\
0_{4} & 0_{4}
\end{array}\right)\otimes \bb I_N
\end{equation}
where, for the first generation,  $M_{u}$ 
contains the Yukawa coupling of the up quark
 and the Dirac  mass of $\nu_e$, $M_{d}$ 
contains the down quark and  the electron masses, 
and $M_{R}$ the
Majorana mass of $\nu_e$. The structure is
repeated for the other two generations.

The real structure
\begin{equation}
J = \mathcal J \otimes J_F\label{eq:46}
\end{equation}
acts as the charge conjugation
operator ${\cal J}= i\gamma^0\gamma^2 cc$
on $L^2(\M,S)$, and as 
\begin{equation}
J_{F}:=\left(\begin{array}{cc}
0 & \mathbb{I}_{16N}\\
\mathbb{I}_{16N} & 0
\end{array}\right) cc
\end{equation}
on $\mathcal H_F$, where it exchanges the blocks ${\cal H}_R\oplus
{\cal H}_L$ of particles with the block ${\cal H}_R^c\oplus
{\cal H}_L^c$ of antiparticles. 
The grading is 
\begin{equation}
\label{graduation}
\Gamma = \gamma^5 \otimes \gamma_F\quad \text{ where } \quad  \gamma_{F}:=\left(\begin{array}{cccc}
\mathbb{I}_{8N}\\
 & -\mathbb{I}_{8N}\\
 &  & -\mathbb{I}_{8N}\\
 &  &  & \mathbb{I}_{8N}
\end{array}\right).\end{equation} 

 The operators $\gamma_F, J_F$ and $D_F$ are such
   that $J_F^2 =\bb I$, $J\gamma_F = - \gamma_F J_F$, $J_F D_F = D_F J_F$, meaning that the finite part of the
   spectral triple of the standard model has $KO$-dimension $6$
   \cite{Barrett:2007vf,Chamseddine:2007oz}.  Meanwhile the continuous part of the spectral triple
   has $KO$-dimension $4$, that is ${\cal J}^2=-\bb I$, ${\cal  J}\gamma^5
   = \gamma^5\cal J$ and ${\cal J}\ds = \ds{\cal J}$.
\bigskip

Gauge fields are obtained by fluctuating the operator $D$ by $\A$, that is
substituting it with the \emph{covariant Dirac operator}
\begin{equation}
  \label{eq:30}
  D_A:= D+ A + JA J^{-1}
\end{equation}
where 
\begin{equation}
  \label{eq:36}
  A = \sum_i a_i [D, b_i]\quad a_i, b_i\in \A
\end{equation}
is a selfadjoint $1$-form of
the almost commutative manifold.

As stressed in the introduction, the field $\sigma$ cannot be
generated by a fluctuation of the Majorana part
\begin{equation}
D_M:= \gamma^5\otimes D_R
\label{eq:23}
\end{equation}
of the Dirac operator, because $[D_R, a]=0$ for any $a\in
  \A_{sm}$. The obstruction has its origin in the first-order
  condition. Indeed one easily checks \cite{Devastato:2013fk} that for any $a,b\in
C^\infty(\M)\otimes\A_F$
\begin{equation}
  \label{eq:123}
  [[D_M, A], \, J b^* J^{-1}] = 0 \text{ if and only if } [D_M, A]=0.
\end{equation}
Hence the necessity to make the first-order condition more flexible
\cite{Chamseddine:2013fk}, or to enlarge the algebra one is starting
with, in order to have enough space to generate the field
$\sigma$ without violating the first-order condition. This enlargement
is made possible by mixing the internal degrees of freedom of ${\cal
  H}_F$ with the spinorial degrees of freedom of $L^2(\M, S)$. This
has been done in \cite{Devastato:2013fk} and is recalled in the next
two paragraphs. 

\subsection{Mixing of spinorial and internal degrees of freedom} 
\label{mixing}

The total Hilbert space ${\cal H}$ of the almost commutative geometry (\ref{eq:02}) is the tensor product of four
dimensional spinors by the 96-dimensional elements of $\mathcal
H_F$. Any of its element is a $\mathbb C^{384}$-vector valued
function on $\M$. From now on, we work with $N=1$ generation only, and
consider instead $384/3 = 128$ components vector. The total Hilbert space
can thus be written - at least in a local trivialization - in two ways:
\begin{equation}
\HH  = L^2(\M, S)\otimes \mathcal H_F \; \text{ or }\; \HH = L^2(\M)\otimes \mathsf H_F
\label{eq:10}
\end{equation}
where $\sf H_F\simeq \mathbb C^{128}$ takes  into account both external (i.e. spin) and internal (i.e. particle) degrees of
freedom. We label
 the basis of $\sf H_F$ with  a multi-index $s\dot s\sC\sI\alpha$
 where:
\begin{itemize}
\item[$\spinind s ,\dotspinind s$]  are
  the  four spinor indices: $\spinind{s=r,l}$ runs over the right, left
  parts  and $\dotspinind{s=\dot 0, \dot 1}$ over
  the particle, antiparticle parts of the spinors. 

\item[$\partind {C}$] indicates whether we are considering
  ``particles'' ($\partind{C=0}$) or ``antiparticles''  ($\partind{
    C=1}$).

\item[$\sI$] is a ``lepto-color'' index: $\sI = 0$ identifies leptons while $\colorind
  I=1,2,3$ are the three colors of QCD.

\item[$\flavind \alpha$] is the flavor index. It runs over the set $u_R,d_R,u_L,d_L$ when $\colorind{I=1,2,3}$, and $\nu_R,e_R,\nu_L,e_L$ when $\colorind{I=0}$. 
\end{itemize}
On this basis, an element $\Psi$ of $\HH$ has components
$\Psi^{\mathsf C \colorind I}_{\spinind s  \dotspinind s
  \flavind{\alpha}}\in L^2(\M)$. The position of the indices is arbitrary: $\Psi$ evaluated
at $x\in\M$ is a column vector, so all the
indices are row indices. An element $A$ in ${\cal B}(\cal H)$ is a
$128\times 128$ matrix whose coefficients are function of $M$, and
carries the indices
\begin{equation}
  \label{eq:97}
  A = A_{\sD s \sJ \dot s \alpha}^{\sC t \sI \dot t \beta}
\end{equation}
where $\sD, t, \sJ, \dot t, \beta$ are column indices with the same
range as $\sC, s, \sI, \dot s, \alpha$.

This choice of indices yields the chiral basis for the Euclidean
Dirac matrices:{\footnote{The multi-index $st$ after the closing parenthesis is
      to recall that the block-entries of the $\gamma$'s matrices
      are labelled by indices $s,t$ taking values in the set
      $\left\{l,r\right\}$. For instance the $l$-row, $l$-column block
      of $\gamma^5$ is $\mathbb I_2$. Similarly the entries of the
      $\sigma$'s matrices are labelled by $\dot s, \dot t$ indices
      taking value in the set $\left\{\dot 0, \dot 1\right\}$: for
      instance ${\sigma^2}^{\dot 0}_{\dot 0} ={\sigma^2}^{\dot 1}_{\dot 1} = 0$. }}
\begin{equation}
  \label{eq:3}
\gamma^\mu=\left(\begin{array}{cc}  0_2
  & {\sigma^\mu} \\{\tilde\sigma}&  0_2
\end{array}\right)_{st},\quad
\gamma^5=\gamma^1\gamma^2\gamma^3\gamma^0 =\left(\begin{array}{cc}  \bb I_2
  & 0_2\\ 0_2& - \bb I_2
\end{array}\right)_{st},
\end{equation}
where for $\mu = 0,1,2,3$ one defines
\begin{equation}
  \label{eq:1}
  \sigma^\mu =\left\{\bb I _2, -i\sigma_i, \right\},\quad
  \tilde\sigma^\mu =\left\{\bb I_2, i\sigma_i\right\}
\end{equation}
with $\sigma_{i}$, $i=1,2,3$ the Pauli matrices. Explicitly,
\begin{equation}
\nonumber
 \label{eq:43}
\sigma^0 =\bb I_2,\; \sigma^1 = - i \sigma_1 =  \left(\begin{array}{cc} 0 & -i \\ - i & 0
\end{array}\right)_{\dot s \dot t},\;
\sigma^2 = - i \sigma_2 =  \left(\begin{array}{cc} 0 & -1 \\ 1 & 0
\end{array}\right) _{\dot s \dot t},\;
\sigma^3 = - i \sigma_3 =  \left(\begin{array}{cc} -i & 0 \\ 0 & i
\end{array}\right) _{\dot s \dot t}.
\end{equation}

 The free Dirac operator $\slashed\del$ extended  to ${\cal H}$
 according to \eqref{eq:89} acts as {\footnote{We use Einstein summation on alternated up/down
    indices. For any $n$ pairs of indices $(x_1,y_1)$, $(x_2, y_2)$, ...
    $(x_n, y_n)$, we write
$\delta_{x_1x_2... x_n}^{y_1y_2... y_n}$ instead of
$\delta_{x_1}^{y_1}\delta_{x_2}^{y_2}... \delta_{x_n}^{y_n}$. For the
tensorial notation to be coherent, $\ds$ and $\gamma^\mu$ should carry lower
$s\dot s$ and upper $t \dot t$ indices. We systematically omit them to
facilitate the reading.}} 
 \begin{equation}
   \label{eq:3-bis}
 \slashed D:=\,\delta^{\sC\sI\beta}_{\sD\sJ\alpha}\slashed \del =
 -i\left(\begin{array}{cc}\delta^{\sI\beta}_{\sJ\alpha}\,\gamma^\mu\nabla_\mu^S&
     0_{64}\\0_{64}&
 \delta^{\sI\beta}_{\sJ\alpha}\,\gamma^\mu\nabla_\mu^S\end{array}\right)_{\sC\sD}.
 \end{equation}
In tensorial notation, the charge conjugation
operator is
\begin{equation}
{\cal J}= i\gamma^0\gamma^2 cc=i\left(
\begin{array}{cc} 
  {{\tilde\sigma}^2}&0_2\\ 
  0_2 & {\sigma^2}
\end{array}\right)_{s t} \, cc = -i\eta_s^t \tau_{\dot s}^{\dot t}\, cc,
\label{eq:2}
\end{equation}
while
\begin{equation}
J_{F}=\left(\begin{array}{cc}
0 & \mathbb{I}_{16}\\
\mathbb{I}_{16} & 0
\end{array}\right)_{\sC\sD} cc,
\end{equation}
hence
\begin{equation}
(J\Psi)^{\partind C\colorind I}_{s\dotspinind s
  \flavind{\alpha}} = - i \eta^t_s\,\tau^{\dot t}_{\dot s}\, \xi^\sC_\sD \, \delta^{\sI\beta}_{\sJ\alpha}\,
\bar \Psi^{\partind D \colorind J
}_{t \dotspinind t \flavind{\beta}}
\label{eq:12}
\end{equation}
where for any pair of indices $x,y \in [1,..., n]$ one defines
\begin{equation}
  \label{eq:22}
  \xi^x_y=\left(\begin{array}{cc} 0_n & \bb I_n \\ \bb I_n &
    0_n \end{array}\right),\quad \eta^x_y=\left(\begin{array}{cc} \bb
    I_n & 0_n \\ 0_n&
   - \bb I_n \end{array}\right), \quad \tau^x_y=\left(\begin{array}{cc} 0_n & -\bb I_n \\ \bb I_n &
    0_n \end{array}\right).
\end{equation}
The chirality acts as $\gamma^5 = \eta_s^t \delta_{\dot s}^{\dot t}$
on the spin indices, and as $\gamma_F = \eta^\sC_\sD\,\delta^\sI_\sJ\,\eta_\alpha^\beta $  on the internal indices:
\begin{equation}
\label{tensorJgamma}
(\Gamma\Psi)^{\partind C\colorind I}_{s\dotspinind s
  \flavind{\alpha}} = \eta_s^t \delta_{\dot s}^{\dot
t}\, \,\eta^\sC_\sD\,\delta^\sI_\sJ\,\eta_\alpha^\beta \; \Psi^{\partind D\colorind J}_{t\dotspinind t
  \flavind{\beta}}.
\end{equation}
%%%%%%%%%%%%%%%%%
\subsection{The grand algebra\label{se:tripleforsm}}

Under natural assumptions (irreducibility of the representation,
existence of a separating vector), a ``symplectic hypothesis'' and the requirement that the $KO$-dimension is $6$,
the most general finite algebra that satisfies the conditions for the
real structure is \cite{Chamseddine:2008uq}
\begin{equation}
\mathcal{\mathcal{A}}_F=M_{a}(\mathbb{H})\oplus M_{2a}(\mathbb{C})
\quad a\in\mathbb{N},
 \label{genericalgebra}
\end{equation}
acting on a Hilbert space of dimension $2(2a)^2$. To have a non-trivial grading on $M_{a}(\mathbb{H})$ the integer 
$a$ must be at least 2, meaning the simplest possibility is 
$
M_{2}(\mathbb{H})\oplus M_{4}(\mathbb{C}). 
\label{bigalgebra}
$
The dimension of the Hilbert space is thus $2(2\cdot2)^2 = 32$, which
is precisely the dimension of $\cal H_F$ for one generation. The grading condition $[a,\Gamma]=0$ imposes the reduction to 
the left-right algebra,
\begin{equation}
\mathcal{A}_{LR}:=\mathbb{H}_L\oplus\mathbb{H}_R\oplus M_{4}(\mathbb{C}),
 \label{repa2}
\end{equation}
and the
order one condition $[[D_F, a], \, Jb^*J^{-1}]=0$ reduces further the algebra to 
${\cal A}_{sm}$ in \eqref{eq:115}.

The case $a=3$ requires an Hilbert space of dimension
$2(2\cdot3)^2= 72$, which has no obvious physical interpretation so
far. 

For $a=4$,
the dimension is $2(2\cdot 4)^2= 128$, which turns out to be precisely
the dimension of the ``fermion doubled'' space $\mathsf H_F$. In other
terms, the mixing of the internal and the spin degrees of freedom provides
exactly the space required to represent the ``grand  algebra''
\begin{equation}
{\cal A}_G=M_4({\mathbb H})\oplus M_8({\mathbb C}).
\end{equation}
Any elements of $\A_G$
is seen as a pair of $8\times 8$ complex matrices $Q\in M_4(\bb H), M\in
 M_8(\mathbb C)$, each having a block
structure of four  $4\times 4$ matrices
\begin{equation}
    \label{eq:16-01}
  Q
= \left(
\begin{array}{cc}
Q_{1}^{1} & Q_{1}^{2}\\
Q_{2}^{1} &Q_{2}^{2}
\end{array}\right) , \quad   M
= \left(
\begin{array}{cc}
 M^{1}_{1} &  M^{2}_{1}\\
M^{1}_{2} & M^{2}_{2}
\end{array}\right)
  \end{equation}
where $Q_i^j\in M_2(\bb H)$ and $M_i^j\in M_4(\bb C)$ for any $i,j
=1,2$. By further imposing all the conditions defining a
spectral triple, one intends to find back the algebra $\A_{sm}$ of the
standard model acting suitably on $\mathcal H_F$. This imposes that
$Q$ acts on the particle subspace $\sC=0$,
 trivially on the
lepto-color index $\sI$, meaning the
complex components of each of the four $4\times 4$ matrices $Q_i^j$ are labelled by
the flavor index $\alpha$. Similarly, one asks that $M$ acts on
antiparticles $\sC=1$, trivially on the flavor
index, meaning the components of each of the four $M_i^j$ are labelled
by the  lepto-color index~$\sI$.  
Thus any element $(Q, M)\in\A_G$ acts on
$\mathsf H_F$ as
\begin{equation}
  \label{eq:139}
\delta_{\mathsf C\sI}^{0\sJ}\,Q_{i \alpha}^{j\beta} \, +\,
\delta_{\mathsf C\alpha}^{1\beta}\, M_{i\,\sI}^{j\sJ}.
\end{equation}
The representation of $\cinf\otimes\A_G$ is obtained
viewing $Q_{i \alpha}^{j\beta}$, $M_{i\,\sI}^{j\sJ}$ no longer as constants but as $L^2$
functions on $\M$. 

There is still some freedom on how  to label the blocks of
the matrices $Q$ and $M$. One simply needs indices $i,j$  that
live on the $s\dot s$ spinorial space, take two
values each and are compatible
with the order-zero condition~\eqref{eq:138}. The natural choice is to
label the blocks of either $Q$ or $M$ by the chiral index
$s=r,l$ and the other blocks by the (anti)-particle index $\dot s= \dot 0, \dot
1$ (although in principle one could also consider combinations
of them). In \cite{Devastato:2013fk} we chose to label the quaternions by the
anti-(particle) index and the complex matrices by the chiral
index,
\begin{equation}
  \label{eq:49bis}
  Q= Q_{\dot s \alpha}^{\dot t\beta}, \quad M = M_{s\sI}^{t\sJ}.
\end{equation}

\noindent The reduction of $\A_G$ to
  the algebra of the standard model is then obtained as follows
\begin{framed}
\vspace{-.5truecm}

  \begin{eqnarray}
    \label{eq:50}
\mathcal{A}_{G}&=& M_4(\mathbb{H}) \oplus M_8(\bb C) \\[.25truecm]
\nonumber 
&\Downarrow& \hspace{0truecm}\text{grading condition}\\[.25truecm] 
\nonumber
{\cal A}'_G &=& M_2(\bb H)_L\oplus  M_2(\bb H)_R\oplus M_4^l(\bb C) \oplus M_4^r(\bb C)\\[.25truecm]  
\nonumber 
&\Downarrow& \hspace{0truecm}\text{$1^\text{st}$-order for the Majorana-Dirac operator $D_M$}\\ [.25truecm] 
   \nonumber
    \A''_G &=& (\mathbb H_L \oplus \mathbb H'_L \oplus \bb C_R \oplus \bb C'_R) \oplus (\bb C^l\oplus M_3^l(\bb C) \oplus \bb C^r \oplus M^r_3(\bb C))
 \text{with }\bb C_R = \bb C^r=\bb C^l\\[.25truecm] 
\nonumber &\Downarrow&\hspace{0truecm} \text{$1^\text{st}$-order for
  the free Dirac operator $\slashed D$}\\[.25truecm] 
\nonumber
\A_{sm} &=& \bb C \oplus \bb H \oplus M_3(\bb C) \\[-1truecm] 
\vspace{-1.75truecm}
\end{eqnarray}
\end{framed}
The interest of the grand algebra is the 
possibility to generate the
field $\sigma$ thanks to a fluctuation of the Majorana mass term
$D_M$ \eqref{eq:23} which
respects the first-order condition imposed by this same Majorana mass
term. Namely, and  this has to be put in contrast
with \eqref{eq:123}, one has that  \cite{Devastato:2013fk} 
\begin{equation}
\text{ for } A\in\A''_G, \; [D_M, A] \text{ is not
  necessarily zero.}\label{eq:82}
  \end{equation}
\subsection{Unboundedness of the commutator}
\label{unboundsection}

The breaking $\A_G\to\A'_G\to\A''_G$ deals with
the finite
dimensional part of the spectral triple. The final breaking
$\A''_G\to\A_{sm}$ is driven by the free Dirac operator and requires
the product with the manifold. However, as explained in \cite{Devastato:2014fk}, there is no spectral triple
for $\cinf\otimes\A_G$ (nor with $\A'_G$ or $\A''_G$) because the commutator $[\slashed D, A]$ of any
of its element with the free Dirac operator is 
never bounded. This can be seen from eq. (5.3) in
\cite{Devastato:2013fk} and has been pointed out to us by
W. v. Suijlekom. In order to have bounded
commutators, the action of $\A_G$ on spinors has to~be~trivial. 

\begin{prop}
\label{nogo}
 Let ${\mathsf A}_F$ be a finite dimensional algebra acting on the Hilbert space
 ${\mathsf H}_F$ in~\eqref{eq:10}. For any $A\in \cinf\otimes {\mathsf
   A}_F$, the
 commutator $[\slashed D, A]$ is bounded iff
 ${\mathsf A}_F$ acts as the identity operator on
the spinor indices $s \dot s$. In particular, the biggest sub-algebra of
 $\cinf \otimes \A_G$ acting as in \eqref{eq:139} and whose commutator
 with $\slashed D$ is bounded is
 $\cinf\otimes(M_2(\bb H) \oplus M_4(\bb C))$. 
 \end{prop}
 \begin{proof}
   In tensorial notation, a generic element of ${\mathsf A}_F$ is $A=A_{\sD
     s\sJ\dot s\alpha}^{\sC t\sI \dot t\beta}$.  For any such $A$,
by \eqref{eq:3-bis} and omitting  the indices $st\dot s\dot t$ for
the Dirac matrices, one gets
   \begin{equation}
     \label{eq:118}
     [\slashed D, A] =
     [\delta^{\sC\sI\beta}_{\sD\sJ\alpha}\slashed\del ,
     A_{\sD s\sJ\dot s\alpha}^{\sC t\sI \dot t\beta}] =
     -i[\delta^{\sC\sI\beta}_{\sD \sJ \alpha}\gamma^\mu,
     A_{\sD s\sJ\dot s\alpha}^{\sC t\sI \dot t\beta} ]\nabla_\mu^S - i\gamma^\mu[
    \nabla_\mu^S, A_{\sD s\sJ\dot s\alpha}^{\sC t\sI \dot t\beta}].
   \end{equation}
This is bounded iff the first term in the r.h.s. is zero. The only matrices that commute with all the Dirac matrices are the
multiple of the identity, hence $[\slashed D, A]$ is bounded iff $A= \lambda \delta_{s\dot s}^{t\dot t} A^{\sC
  \sI\beta}_{\sD\sJ\alpha}$ for some scalar $\lambda$. This means
$Q^{j\beta}_{i\alpha}=\lambda\delta_{s\dot s}^{t\dot t}
Q^{\alpha}_{\beta}\in M_2(\bb H)$ and
$M_{i\sI}^{j\sJ} =\lambda\delta_{s\dot s}^{t\dot t} M_\sI^\sJ\in M_4(\bb C)$ in \eqref{eq:139}.
\end{proof}

In other terms, in order to build a spectral triple with the
grand algebra (that is $a=4$ in \eqref{genericalgebra}), one has to consider its subalgebra given by
$a=2$, that acts without mixing spinorial and internal indices. This
of course is not interesting
from our perspective, since the aim of the grand  algebra is precisely to mix spinorial with internal degrees
of freedom. A solution is to consider instead \emph{twisted spectral
triples}. They have been introduced in \cite{Connes:1938fk} precisely
to solve the problem of the unboundedness of the commutator, which
may occur in very elementary situations such as the lift to spinors of a
conformal transformation. Using twists to make $[\slashed D, A]$  bounded has been
suggested independently to the second author by J.-C. Wallet, and to
the first author by W. v. Suijlekom, who also brought our attention to ref.\cite{Connes:1938fk}.
\bigskip

\section{Twisting the standard model}\setcounter{equation}{0}
\label{section:twisting}

A \emph{twisted spectral triple}
is a triple $(\A, {\cal H}, D)$ where $\A$ is an involutive algebra
acting on a Hilbert space $\cal H$ and $D$ a selfadjoint operator on
$\cal H$ with compact resolvent,  together with an automorphism $\rho$ of $\A$ such that
\begin{equation}
  \label{eq:84}
  [D, a]_\rho = Da -\rho(a)D
\end{equation}
is bounded for any $a\in \A$. It is graded if, in addition, there is
a selfadjoint  operator $\Gamma$ of square $\bb I$ which commutes the algebra and
anticommutes with $D$.

As far as we know, the other conditions satisfied by a spectral triple
have not been adapted to the twisted case yet. As long as the
commutator between the algebra and the Dirac operator is not involved,
 one can keep the definitions
of an ordinary spectral triple,  for instance the order-zero
condition. In the $1^{\text{st}}$-order
condition \eqref{eq:121} it is natural to
substitute $[D,a]$ with the twisted commutator $[D,
a]_\rho$. The question is whether to twist the commutator with $J b^*
J^{-1}$ as well. As explained in \cite[Prop. 3.4]{Connes:1938fk}, the set $\Omega^1_D $ of twisted $1$-forms, that is
all the operators of the form
\begin{equation}
  \label{eq:167}
 \mathbb A = \sum_i b^i [D, a_i]_\rho,
\end{equation}
is a $\A$-bimodule for the left and right  actions
\begin{equation}
  \label{eq:182}
  a\cdot\omega\cdot b := \rho(a) \omega b \quad \forall a,b\in\A,
  \omega\in \Omega^1_D.
\end{equation}
Therefore it is natural to
twist the commutator with $JbJ^{-1}$. As pointed out below
proposition \ref{lemmasigma},  this choice
is also the one which is efficient for our purposes. Furthermore we
assume that $\rho$ is a $*$-automorphism that commutes with the real
structure $J$ (more on that matter is discussed in the conclusion), which permits to define the twisted version of the
$1^\text{st}$-order condition as follows.
\begin{df}
\label{deftwistedfirstorder}
 A twisted spectral triple $({\cal A}, {\cal H}, D, \rho)$ with real
 structure $J$ satisfies the twisted $1^\text{st}$-order condition if and
 only if
 \begin{equation}
   \label{eq:112}
  [ [D, a]_\rho,\, JbJ^{-1}]_\rho =    [D, a]_\rho\, JbJ^{-1} - J
  \rho(b) J^{-1}  [D, a]_\rho =  0 \quad\quad \forall a, b\in \A.
 \end{equation}
\end{df}
\noindent When $\rho=\rho^{-1}$ (which will be the case here), requiring that
$\rho$ is a $^*$-automorphism is equivalent to the unitary condition
(3.4) of \cite{Connes:1938fk}.
 
\subsection{Representation}
\label{subsec:representation}

For reasons explained in \S~\ref{secdisc}, it
is convenient to work with the other natural
representation of the grand algebra than the one used in
\cite{Devastato:2013fk}. Namely instead of
\eqref{eq:49bis} one asks that quaternions carry the
chiral index $s$ of spinors while the complex matrices carry the (anti)-particle
index:
\begin{equation}
  \label{eq:49}
  Q= Q_{ s \alpha}^{t\beta}, \quad M = M_{\dot s\sI}^{\dot t\sJ}.
\end{equation}
 Explicitly, the representation of the grand
algebra $\A_G$ is 
  \begin{equation}
    \label{eq:16}
Q = \left(
\begin{array}{cc}
 Q^{r}_{r} &  Q^{l}_{r}\\
Q^{r}_{l} & Q^{l}_{l}
\end{array}\right)_{s t}\in M_4(\bb H),\quad
  M%_{\dot s\alpha}^{\dot t\beta}
= \left(
\begin{array}{cc}
M_{\dot 0}^{\dot 0} & M_{\dot 0}^{\dot 1}\\
M_{\dot 1}^{\dot 0} &M_{\dot 1}^{\dot 1}
\end{array}\right)_{\dot s \dot t} \in M_8(\bb C) , 
  \end{equation}
where for any
$s,
t\in\left\{l,r\right\}$ and $\dot s, \dot t\in\left\{\dot 0, \dot 1\right\}$ one defines
\begin{equation*}
Q_s^t =  \left(\begin{array}{cccc}
 Q_{ s a}^{t a}  & Q_{s a}^{t b}  & Q_{s a}^{t c}& Q_{s a}^{ t d}\\
  Q_{ s b}^{t a}  & Q_{ s b}^{t b}  & Q_{s b}^{t c}& Q_{s b}^{t d}\\
 Q_{s c}^{t a}  & Q_{s c}^{t b}  & Q_{s c}^{t c}& Q_{s c}^{t d}\\
 Q_{s d}^{t a}  & Q_{s d}^{t b}  & Q_{s d}^{t c}& Q_{s d}^{t d}
\end{array}\right)_{\alpha\beta}\!\!\!\!\!\!\!\!\in M_2(\bb H),\quad
  M_{\dot s}^{\dot t} = \left(\begin{array}{cccc}
 M_{\dot s 0}^{\dot t 0}  & M_{\dot s 0}^{\dot t 1}  & M_{\dot s 0}^{\dot t 2}& M_{\dot s 0}^{\dot t 3}\\
  M_{\dot s 1}^{\dot t 0}  & M_{\dot s 1}^{\dot t 1}  & M_{\dot s 1}^{\dot t 2}& M_{\dot s 1}^{\dot t 3}\\
 M_{\dot s 2}^{\dot t 0}  & M_{\dot s 2}^{\dot t 1}  & M_{\dot s 2}^{\dot t 2}& M_{\dot s 2}^{\dot t 3}\\
 M_{\dot s 3}^{\dot t 0}  & M_{\dot s 3}^{\dot t 1}  & M_{\dot s 3}^{\dot t 2}& M_{\dot s 3}^{\dot t 3}
\end{array}\right)_{\sI\sJ}\!\!\!\!\in \! M_4(\bb C).
\end{equation*}
Here we use $a,b,c,d$ to denote the value of the flavor
index $\alpha$.  On the remaining indices, $Q$ and $M$ act trivially,
that is as the
identity operator. The representation of  $A=(Q,M)\in{\mathcal A}_G$
on $\mathsf H_F$ is thus
\begin{equation}
A^{\sC t\sI \dot t \beta}_{\sD  s \sJ \dot s \alpha} = \left(
\delta^{\sC \dot t\sI}_{0  \dot s\sJ} \,Q^{ t \beta}_{ s \alpha}
+
\delta^\sC_1 M^{\dot t\sI}_{\dot s\sJ}\delta_{s\alpha}^{t\beta}\right) 
=\left(
\begin{array}{cc} \delta^{\dot t\sI}_{\dot s\sJ} \,Q^{t \beta}_{s
    \alpha} & 0_{64}\\ 0_{64} &  M^{\dot t\sI}_{\dot s\sJ}\, \delta_{s\alpha}^{t\beta}\end{array}\right)_{\sC\sD}. 
\label{repa4}
\end{equation}

One easily checks the order-zero condition \eqref{eq:138}: with
$A=(R, N)\in {\cal A}_G$, a generic element of the opposite algebra is
\begin{equation}
  \label{eq:29}
  J A J^{-1}= -J  A J=  \left(\begin{array}{cc}
-\delta_{s\alpha}^{t\beta} \, (\tau\bar
   N\tau)_{\dot s\sJ}^{\dot t\sI}& 0_{64}\\
0_{64} & \delta_{\dot s\sJ}^{\dot t\sI}\, (\eta \bar R\eta)_{s\alpha}^{ t \beta}
\end{array}\right)_{\sC\sD}
\end{equation}
where the bar denotes the complex conjugate and we used
\begin{equation}
   \label{eq:31}
 {\cal J} R {\cal J} := (\tau^2)_{\dot s}^{\dot t} \, (\eta \bar
 R\eta)_{s\alpha}^{ t \beta} = - \delta_{\dot s}^{\dot t}\, (\eta \bar R\eta)_{s\alpha}^{ t \beta} ,\; \quad
 {\cal J} N {\cal J} : = (\eta^2)_s^t\,(\tau\bar
   N\tau)_{\dot s\sJ}^{\dot t\sI}= \delta_s^t (\tau\bar
   N\tau)_{\dot s\sJ}^{\dot t\sI}.
 \end{equation}
Obviously \eqref{repa4} commutes with \eqref{eq:29}.

\begin{lemma}
The biggest
  subalgebra of $\cinf \otimes\A_G$ that satisfies the grading condition 
  \eqref{eq:136} and has bounded commutator with $\slashed D$ is the
  left-right algebra  $\A_{LR}$ given in \eqref{repa2}. 
\end{lemma}
\begin{proof}
By \eqref{tensorJgamma}, for
the quaternion sector $[\Gamma, A]=0$ amounts to asking 
 $[\eta_s^t\eta_\alpha^\beta,Q_{s\alpha}^{t\beta}]=0$. This  imposes 
\begin{equation}
  \label{eq:58}
  Q = \left( 
    \begin{array}{cc}
      Q^r_r & 0_4 \\ 
      0_4 & Q^l_l
    \end{array}\right)_{st} 
\end{equation}
where 
\begin{equation}
   \label{eq:60}
   Q^{r}_{r}=\left(
     \begin{array}{cc}
       {q^r_R}& 0_2 \\
       0_2 & {q^r_L}
     \end{array}
  \right)_{\alpha\beta}, \; Q^{l}_{l}=\left(
     \begin{array}{cc}
       {q^l_R}& 0_2 \\
       0_2 & {q^l_L}
     \end{array}
  \right)_{\alpha\beta} \quad
\text{ with } q^r_R, q^r_L, q^l_R, q^l_L\in \bb H.
\end{equation}
For matrices, one asks $[\delta^{\dot t \sI}_{\dot s
  \sJ}, M^{\dot t \sI}_{\dot s \sJ}]=0$ which is trivially
satisfied. So the grading condition $[\Gamma,
A]=0$ imposes the reduction of $\A_G$ to
\begin{equation}
  \label{eq:62}
{\cal  B}_{LR} :=(\bb H^l_L \oplus \bb H^r_L \oplus \bb H^l_R \oplus \bb H^r_R) \oplus
M_8(\bb C). 
\end{equation}

For $A=(Q, M)\in \cinf\otimes{\cal B}_{LR}$, the boundedness of the commutator
{\footnote{\label{footnotenotations}To lighten notation, we omit the trivial indices in the
    product (hence in the commutators) of operators. From
    \eqref{eq:49} one knows that $Q$ carries the indices
    $s\alpha$ while $\gamma^\mu$ carries $s\dot s$, hence
    $[\slashed \del, Q]$ carries indices $s\dot s\alpha$ and should be
    written $[\delta_\alpha^\beta\slashed\del,  \delta_{\dot s}^{\dot t} Q]$.
As well, $[\slashed \del, M]$ carries indices $s\dot s I$ and holds for $[\delta_I^J\slashed \del, \delta_{s}^{t} M]$.}}
\begin{equation}
\label{eq:41}
   [\slashed D, A] =\left(
\begin{array}{cc}
 \delta^I_J \,[\slashed\del,  Q]&0_{64} \\ 
 0_{64} &\delta_\alpha^\beta\, [\slashed \del, M]
\end{array}
\right)_{\partind{CD}}
 \end{equation}
means that 
\begin{equation}
[\ds, Q] = -i\gamma^\mu (\nabla_\mu^S Q)  -i [\gamma^\mu, Q]\nabla_\mu^S \quad\text{ and }\quad
[\ds, M] = -i\gamma^\mu (\nabla_\mu^S M) -i[\gamma^\mu, M]\nabla^S_\mu
\label{eq:137}
\end{equation}
are
bounded, where $(\nabla_\mu^s Q):= (\del_\mu Q) + [\omega_\mu^S, Q]$ and
similarly for  $(\nabla_\mu^s M)$. This is obtained if and only if $Q$ and $M$ commute with all the Dirac
matrices, i.e. are proportional to $\delta_{s\dot s }^{t\dot t}$. For $Q$ this means 
$Q_r^r = Q_l^l$ in \eqref{eq:58}, hence the reductions
\begin{equation}
\bb H^r_R \oplus \bb H^l_R\to \bb H_R, \quad \bb H^r_L \oplus \bb H^l_L \to
\bb H_L.
\label{eq:94}
\end{equation}
For $M$, this means that all the components $M_{\dot s}^{\dot t}$ in
\eqref{eq:16} are equal, that is the reduction
\begin{equation}
\label{reduc}
M_8(\bb
C)\to M_4(\bb C).
\end{equation}
Therefore ${\cal B}_{LR}$ is reduced to $\A_{LR}$, acting diagonally on spinors.
\end{proof}
\noindent This lemma is nothing but a restatement of  Prop. \ref{nogo} in
the peculiar representation \eqref{repa4} and taking into account the
grading condition. Nevertheless, it is useful to have it explicitly, in
order to understand how to get rid of the unboundedness of the commutator.
It is also worth stressing the difference with the representation
\eqref{eq:49bis}, for which the grading breaks both matrices and
quaternions and reduces $\A_G$ to $\A'_G$. Here only quaternions are
broken by the grading.

To cure the unboundedness of the  commutator,  the idea we propose 
is the following: impose the reduction
\eqref{reduc} by hand, and deal with the unboundedness of $[\ds, Q]$
thanks to a twist.  This is a ``middle term
solution'': imposing by hand both reductions \eqref{reduc} and \eqref{eq:94}
is not interesting from the grand algebra point of view, since it brings
us back to an almost commutative geometry where spinorial and internal
indices are not mixed; solving both the unboundedness of $[\ds, Q]$ and
$[\ds, M]$ by a twist yields some complications discussed in \S\ref{secdisc}. The
remarkable point is that this middle term solution is sufficient to obtain the
$\sigma$-field by a fluctuation that respects the twisted first-order
condition of definition \ref{deftwistedfirstorder}.

\subsection{The twist and the first-order condition
  for the free Dirac operator}
\label{subsec:twistedfirtsfree}

Imposing \eqref{reduc} on the grand algebra $\A_G$ reduced
by the grading to ${\cal B}_{LR}$ yields
\begin{equation}
  \label{eq:87}
  {\cal B}' := (\bb H^l_L \oplus \bb H^r_L \oplus \bb H^l_R \oplus \bb H^r_R) \oplus
M_4(\bb C).
\end{equation}
An element
$A=(Q,M)$ of $\cal B'$ is given by \eqref{repa4} where $Q$ is as in \eqref{eq:58} while 
$M$ in \eqref{eq:49} is proportional to $\delta_{\dot s}^{\dot t}$:
\begin{equation}
  \label{eq:92}
  M= \delta_{\dot s}^{\dot t} M^\sI_\sJ \in M_4(\bb C).
\end{equation}
The algebra $\cal B'$ contains the algebra of
the standard model $\A_{sm}$, and still has a part (the quaternion)
that acts in a non-trivial way on the spin degrees of freedom. In this sense $\cal B'$ is still from the
grand algebra side, even if it is ``not so grand''.

\medskip Let $\rho$ be the automorphism of $(\bb H^l_L \oplus \bb H^r_L \oplus \bb H^l_R \oplus \bb H^r_R)$  that
exchanges $Q^r_r$ and $Q^l_l$ in \eqref{eq:58},  that is the exchange
\begin{equation}
\bb H^r_R \leftrightarrow \bb H^l_R,\quad \bb H^r_L
\leftrightarrow\bb H^l_L.\label{eq:12402}
\end{equation}
This means in components
 \begin{equation}
 \label{eq:5800}
 \rho\left(
   \left( 
     \begin{array}{cc}
       Q^r_r & 0_4 \\ 
       0_4 & Q^l_l
     \end{array}\right)_{st}\right) 
 = \left( 
   \begin{array}{cc}
     Q^l_l & 0_4 \\ 
     0_4 & Q^r_r
   \end{array}
 \right)_{st}.
 \end{equation}

\begin{lemma} 
\label{lemautomorfismo} Denote by the same letter the extension of $\rho$ to
  $\cinf\otimes (\bb H^l_L \oplus \bb H^r_L \oplus \bb H^l_R \oplus
  \bb H^r_R)$.
For any $\mu$ one has
\begin{equation}
  \label{eq:130}
  \gamma^\mu Q   = \rho(Q) \gamma^\mu,\quad \gamma^\mu\rho(Q)= Q\gamma^\mu.
\end{equation}
Thus
\begin{equation}
 [\ds,  Q]_\rho = -i\gamma^\mu (\nabla^S_\mu Q).
\label{eq:78}
 \end{equation}
\end{lemma}
\begin{proof}
 Writing explicitly the $\delta$'s, one gets
 \begin{align}
\gamma^\mu Q  &= \left(\delta_\alpha^\beta\left( \begin{array}{cc} 0_2 & \sigma^\mu \\
    \tilde\sigma^\mu & 0_2 \end{array}\right)_{st} \right) \left(\left( 
     \begin{array}{cc}
       Q^r_r & 0_4 \\ 
       0_4 & Q^l_l
     \end{array}\right)_{st} \delta_{\dot s}^{\dot t}\right)
=  \left( 
     \begin{array}{cc}
       0_8&\sigma^\mu Q_l^l \\ 
       \bar\sigma^\mu Q^r_r & 0_8
     \end{array}\right)_{st} \\ &=\left(\left( 
     \begin{array}{cc}
       Q^l_l & 0_4 \\ 
       0_4 & Q^r_r
     \end{array}\right)_{st} \delta_{\dot s}^{\dot t}\right)\left(\delta_\alpha^\beta\left( \begin{array}{cc} 0_2 & \sigma^\mu \\
    \tilde\sigma^\mu & 0_2 \end{array}\right)_{st} \right) = \rho(Q)\gamma^\mu.
\end{align}
The second part of \eqref{eq:130} follows because $\rho^2=\mathbb{I}$.
Eq. \eqref{eq:78} comes from 
\begin{equation*}
  \label{eq:79}
  [\ds, Q]_\rho =  -i\gamma^\mu(\nabla_\mu^S Q) - i[\gamma^\mu, Q]_\rho\nabla_\mu^S,
\end{equation*}
where the second term is zero by \eqref{eq:130}.\end{proof}

We still denote by the same
letter  the extension of $\rho$ to  $\cinf\otimes\cal B'$:
\begin{equation}
  \label{eq:91}
  \rho((Q, M)) := ((\rho(Q), M).
\end{equation}
\begin{prop}
\label{twisted-spec-triple}
  $(\cinf\otimes{\cal B}', \HH, \slashed D, \rho)$ together with the grading
  $\Gamma$ in \eqref{graduation} and the real structure $J$ in
  \eqref{eq:46} is a graded twisted spectral
  triple which satisfies the twisted first-order condition of
  definition \ref{deftwistedfirstorder}. 
\end{prop}
\begin{proof}
Let $A=(Q,M)\in\cinf\otimes {\cal B}'$. The twisted version of \eqref{eq:41} is
\begin{equation}
\label{eq:410}
   [\slashed D, A]_\rho =\left(
\begin{array}{cc}
 \delta^I_J \,[\slashed\del,  Q]_\rho&0_{64} \\ 
 0_{64} &\delta_\alpha^\beta\, [\slashed \del, M]
\end{array}
\right)_{\partind{CD}}. 
 \end{equation}
From \eqref{eq:92} and \eqref{repa4} $M$ commutes with $\gamma^\mu$,
so that the second equation in \eqref{eq:137} reduces to
\begin{equation}
\label{eq:410bis}
[\slashed \del, M]= -i\gamma^\mu (\nabla^S_\mu M),
  \end{equation}
which is a bounded operator. By lemma \ref{lemautomorfismo}, $[\slashed\del,  Q]_\rho = -i\gamma^\mu (\nabla_\mu^S Q)$
is bounded as well. Hence $(\cinf\otimes \cal B)', {\cal H}, \slashed
D, \rho$) together with $\Gamma$ form a graded twisted spectral
triple.  

We now examine the twisted first-order condition \eqref{eq:84}. Let
$B=(R,N)\in C^\infty(\M)\otimes {\cal B}'$. A generic element of the algebra opposite to $C^\infty(\M)\otimes{\cal B}'$ is
\begin{equation}
  \label{eq:299}
  J B J^{-1}= -J  B J=  \left(\begin{array}{cc}
\delta_{s\alpha}^{t\beta} \bar N& 0_{64}\\
0_{64} & \delta_{\dot s \sI}^{\dot t \sJ} \bar R
\end{array}\right)_{\sC\sD}
\end{equation}
where we used \eqref{eq:29} and noticed that for $R$ as in
\eqref{eq:58} and $N$ as in \eqref{eq:92} one has
\begin{equation}
   \label{eq:31-Bis}
 (\eta \bar R\eta)_{s\alpha}^{ t\beta}  = \bar R_{s\alpha}^{ t\beta} ,\; \quad
 (\tau\bar N\tau)_{\dot s\sJ}^{\dot t\sI}= - \bar N_{\dot s\sJ}^{\dot t\sI}.
 \end{equation}
As well, one has
\begin{equation}
  \label{eq:290}
  J \rho(B) J^{-1}= -J  \rho(B)  J=  \left(\begin{array}{cc}
\delta_{s\alpha}^{t\beta} \bar N& 0_{64}\\
0_{64} & \delta_{\dot s \sI}^{\dot t \sJ} \rho(\bar R)
\end{array}\right)_{\sC\sD}.
\end{equation}
Thus $[\slashed D, A]_\rho JBJ^{-1} - J\rho(B)J^{-1}   [\slashed D,
A]_\rho$ is a diagonal matrix with components
\begin{equation}
  \label{eq:119}
  [\delta^I_J \,[\slashed\del,  Q]_\rho,\, \delta_{s\alpha}^{t\beta} \bar
  N],\quad\quad  \delta^\beta_\alpha \,[\slashed\del,  M]\, \delta^{\dot t
    \sI}_{\dot s \sJ}\bar R - \delta^{\dot t
    \sI}_{\dot s \sJ}\rho(\bar R) \,\delta^\beta_\alpha \,[\slashed\del,  M].
\end{equation}
The first term vanishes because the only non-trivial index carries by
$\bar N$ is $\sI\sJ$. The
second term is (omitting the deltas and a global $-i$ factor) 
\begin{align}
\nonumber
&\left(
  \begin{array}{cc} 
    0_8 & \sigma^\mu(\del_\mu M)\\ 
    {\bar \sigma}^\mu(\del_\mu M) & 0_8
  \end{array}
\right)_{st} 
\left(
  \begin{array}{cc} 
    \bar R_r^r & 0_8\\ 
    0_8& \bar R_l^l
  \end{array}
\right)_{st} 
-
\left(
  \begin{array}{cc} 
    {\bar R}_l^l & 0_8\\ 
    0_8& {\bar R}_r^r
\end{array}
\right)_{st} 
\left(
  \begin{array}{cc} 
    0_8 & \sigma^\mu(\del_\mu M)\\ 
    {\bar \sigma}^\mu(\del_\mu M) & 0_8
  \end{array}
\right)_{st}
\\
&  \label{eq:120bis}
=\left(
  \begin{array}{cc} 
  0_8 & [\sigma^\mu(\del_\mu M), \bar{R}_l^l] \\ 
\left[{\bar \sigma^\mu(\del_\mu M)}, \bar{R}_r^r\right] & 0_8 \end{array}\right)_{st} 
\end{align}
which vanishes because $R$'s only non-trivial index is $\alpha\beta$
while $\left[{\bar \sigma^\mu(\del_\mu M)}, \bar{R}_r^r\right]$ is proportional to~$\delta_\alpha^\beta$.
\end{proof}

\subsection{Twisted first-order condition for the Majorana-Dirac operator}
\label{subsec:twistedfirstmaj}

We individuate a subalgebra $\cal B$ of $\cal B'$ such that a twisted
fluctuation of the \emph{Majorana-Dirac operator} $D_M$ in
\eqref{eq:23} by ${\cal B}$ satisfies the twisted first-order
condition. Since we are working with one generation of fermions only,
in \eqref{eq:D_F  Modello Standard-1-1}
the Majorana mass matrix ${\cal M}_R$ is $ \Xi_{\alpha}^{\beta} k_R$, where 
 \begin{equation}
  \label{eq:80}
  \Xi =\left(\begin{array}{cc}
1 &0 \\ 0 &0_3 \end{array}\right)
\end{equation}
denotes the projection on the first component. Therefore
\begin{equation}
\label{majoranadiracoperator}
D_M = \gamma^5 D_R = \eta_s^t\,\delta_{\dot s}^{\dot t} \Xi_{\sJ\alpha}^{\sI\beta} \left(\begin{array}{cc} 0 & k_R \\ \bar k_R& 0\end{array}\right)_{\sC\sD}.
\end{equation}
In this equation the product
$\gamma^5 D_R$ is intended with the convention of the footnote
p.\pageref{footnotenotations}, namely this is the tensorial notation
${\gamma^5}_{s\dot s}^{t\dot t} {D_R}_{\sJ\alpha\sC}^{\beta\sI\sD}$ in which we omit the indices.
In practice, this amounts to omit the tensor product symbol in
$\gamma^5\otimes D_R$, which
makes sense because of our choice of viewing the total Hilbert space
no longer as the tensor product of spinors by $\cal H_F$. These
distinctions may
seem pedantic here, but they will be important later on, when writing the
product $\gamma^\mu X_\mu$ for a vector field $X_\mu$ that 
no longer commutes with the Dirac matrices:  $\gamma^\mu X_\mu$ will
hold for ${\gamma^\mu}_{s\dot s}^{t\dot t}\, {X_\mu}_{\sJ\alpha\sC}^{\beta\sI\sD}$, while $\gamma^\mu\otimes X_\mu$
will no longer make sense.

\begin{prop}
\label{lemmasigma}
A subalgebra of ${\cal B}'$ which satisfies the twisted first-order
condition
\begin{equation}
[[D_M, A]_\rho, JBJ^{-1}]_\rho = 0 
\label{eq:93}
\end{equation}
is 
\begin{equation}
  \label{eq:96}
  {\cal B}:= \bb H^l_L \oplus \bb
  H^r_L\oplus \bb C^l_R \oplus \bb C_R^r \oplus  M_3(\bb C).
\end{equation}
\end{prop}
\begin{proof}
Consider first the subalgebra
\begin{equation}
\tilde{\cal B}:= (\bb H^l_L \oplus \bb H^r_L \oplus \bb C^l_R \oplus \bb C^r_R)\oplus
(M_3(\bb C) \oplus \bb C)
\label{eq:71}
\end{equation}
of $\cal B'$ obtained by asking that $q^l_R, q^r_R$ in
\eqref{eq:60} are diagonal quaternions, namely
\begin{equation}
  \label{eq:103}
  q^l_R=\left(\begin{array}{cc} c^l_R & 0 \\ 0 &\bar
      c^l_R\end{array}\right), \; q^r_R=\left(\begin{array}{cc} c^r_R & 0 \\ 0 &\bar
      c^r_R\end{array}\right)\,\text{ with }
  c^l_R, c^r_R\in\bb C;
\end{equation}
while $M$ in \eqref{eq:92} is of the form
\begin{equation}
  \label{eq:101}
M = \delta_{\dot s}^{\dot t} \left(\begin{array}{cc}
     m & 0 \\ 0 & \bf M \end{array}\right)_{\sI\sJ} \text{ with }
m\in \bb C, {\bf M}\in M_3(\bb C).
\end{equation}
This means that while $Q$ carries non-trivial indices $\dot s,
  \alpha$, the action of $M$ is non-trivial only in the $\sI$ index. 
Define similarly  $B=(R,N)\in\tilde{\cal B}$ with components
  $d^l_R, d^r_R\in\bb C$, $n\in \bb C$, ${\bf N}\in M_3(\bb C)$. 
For any $A, B\in\cal \tilde B$, one has
from~\eqref{majoranadiracoperator}  where we write
    \begin{equation}
\sD_M := \eta_s^t\,\delta_{\dot s}^{\dot t}\,
  \Xi_{\sJ\alpha}^{\sI\beta},
\label{eq:190}
  \end{equation}
and \eqref{repa4} (omitting the deltas)
\begin{equation}
   \label{eq:95}
   [D_M, A]_\rho = \left(\begin{array}{cc}
  0_{64}& k_R(\sD_M M -\rho(Q)\sD_M) \\  \bar k_R(\sD_M Q-M\sD_M) &0_{64} 
 \end{array}\right)_{\sC\sD}.
 \end{equation}
By \eqref{eq:299}, \eqref{eq:290} one obtains
{\small{\begin{equation*}
  \label{eq:99}
    [[D_M, A]_\rho, JBJ^{-1}]_\rho = \!\left(\!\!\!\begin{array}{cc}
  0_{64}&\!\!\!\!\!\!\!\!\!\! k_R\left( (\sD_M M \!-\!\rho(Q)\sD_M) \bar R \!-\!\bar N (\sD_M M
  -\rho(Q)\sD_M) \right)\\[5pt]
 \bar k_R\left((\sD_M Q \!-\! M\sD_M)  \bar N\!\!-\!\! \rho(\bar R) (\sD_M Q
   \!-\! M\sD_M) \right)&\!\!\!\!\!\!\! 0_{64}
 \end{array}\!\!\!\right)_{\sC\sD}.
\end{equation*}}}
The terms entering the upper-right components of this matrix
are (omitting a global $k_R$ factor)
\begin{align}
  \label{eq:102-1}
\bar N {\sf D_M} M & =( \bar N \Xi M)_{\dot s \sJ}^{\dot t \sI} (\eta\Xi)_{s\alpha}^{t\beta}
= \left(\begin{array}{cc}  \bar{\sf n}\sf m  & 0_4 \\ 0_4 &\bar{\sf n}\sf m 
    \end{array}\right)_{\dot s \dot t}\otimes\left(\begin{array}{cc} \Xi & 0_4 \\0_4 &-\Xi\end{array}\right)_{st},\\[5pt]
\bar N\rho(Q) \sf D_M &= (\bar N\Xi)_{\dot  s\sJ}^{\dot
  t\sI}\,(\rho(Q)\eta\Xi)_{s\alpha}^{t\beta} = \left(\begin{array}{cc}  \bar{\sf n}& 0_4 \\ 0_4 &\bar{\sf n} \end{array}\right)_{\dot s \dot t}\otimes\left(\begin{array}{cc}
  {\sf c}^l_R& 0_4 \\0_4 &-\, {\sf c}^r_R\end{array}\right)_{st},\\[5pt] 
 {\sf D_M} M \bar R &=  (\Xi M)_ {\dot s\sJ}^{\dot t\sI}\,(\eta\Xi\bar R)_{s\alpha}^{t\beta} = \left(\begin{array}{cc} \sf m& 0_4 \\0_4 & \sf m
    \end{array}\right)_{\dot s \dot t}\otimes\left(\begin{array}{cc}
     \bar{\sf d}^r_R& 0_4 \\0_4 & -\bar{\sf d}^l_R\end{array}\right)_{st},\\[5pt]
\rho(Q){\sf D}_M \bar R &=(\Xi\delta)_{\dot s\sJ}^{\dot t\sI}  \,(\rho(Q)\eta\Xi R)_{s\alpha}^{t\beta} = \left(\begin{array}{cc} \Xi& 0_4 \\0_4 & \Xi    \end{array}\right)_{\dot s \dot t}\otimes\left(\begin{array}{cc}   {\sf c}^l_R\,\bar{\sf d}^r_R& 0_4 \\0_4 & -{\sf c}^r_R\,\bar{\sf d}^l_R\end{array}\right)_{st},
\end{align}
where we defined
 \begin{equation}
   \label{eq:104}
   {\sf m} :=\left(
\begin{array}{cc}
 m & 0 \\ 
0& 0_3
\end{array}\right)_{\sI\sJ},\quad {\sf c}^r_R =
\left(\begin{array}{cc} c^r_R & 0 \\ 0 &
    0_3\end{array}\right)_{\alpha\beta},\quad
{\sf c}^l_R =
\left(\begin{array}{cc} c^l_R & 0 \\ 0 & 0_3\end{array}\right)_{\alpha\beta}
 \end{equation}
and similarly for ${\sf d}^r_R$, ${\sf d}^l_R$ and $\bf n$. Collecting the various
terms, one finds that the upper-right component of $[[D_M, A]_\rho, JBJ^{-1}]_\rho$ vanishes if and only if
\begin{equation}
  \label{eq:105}
  (c^l_R -m ) (\bar d^r_R - \bar n) = 0,\quad (c^r_R-m)(\bar d^l_R - \bar n) =  0.
\end{equation}

Similarly, for the lower-left component of $[[D_M, A]_\rho, JBJ^{-1}]_\rho$ one has
\begin{align}
  \label{eq:102bis}
 \rho(\bar R) M \sf D_M & =(\Xi M)_{\dot s \sJ}^{\dot t \sI}
 (\rho(\bar R) \eta\Xi)_{s\alpha}^{t\beta} = \left(\begin{array}{cc}  \sf m  & 0_4 \\ 0_4 &\sf m 
    \end{array}\right)_{\dot s \dot t}\otimes\left(\begin{array}{cc}
      {\sf \bar d}^l_R& 0_4 \\0_4 &-{\sf \bar d}^r_R\end{array}\right)_{st},\\[7.5pt]
 \rho(\bar R){\sf D}_M Q&= (\Xi\delta)_{\dot  s\sJ}^{\dot t\sI}\,(\rho(\bar R)\,\eta\Xi \, Q)_{s\alpha}^{t\beta} = \left(\begin{array}{cc}  \Xi& 0_4 \\ 0_4 &\Xi\end{array}\right)_{\dot s \dot t}\otimes\left(\begin{array}{cc}
  \,{\sf c}^r_R \bar{\sf  d}^l_R& 0_4 \\0_4 &-{\sf c}^l_R{\sf \bar d}^r_R\end{array}\right)_{st},\\[7.5pt] 
M{\sf D}_M \bar N &=  (M\Xi\bar N)_ {\dot s\sJ}^{\dot t\sI}\,(\eta\Xi)_{s\alpha}^{t\beta} = \left(\begin{array}{cc} \sf \bar n \sf m& 0_4 \\0_4 & \sf \bar n\sf m
    \end{array}\right)_{\dot s \dot t}\otimes\left(\begin{array}{cc}
      \,\Xi& 0_4 \\0_4 & -\Xi\end{array}\right)_{st},\\[7.5pt]
{\sf D}_M Q\bar N&=(\Xi \bar N)_{\dot s\sJ}^{\dot t\sI}  \,
(\eta \Xi Q)_{s\alpha}^{t\beta} = \left(\begin{array}{cc} \bar{\sf
      n}& 0_4 \\0_4 & \bar{\sf n}  \end{array}\right)_{\dot s \dot
  t}\otimes\left(\begin{array}{cc}{\sf c}^r_R& 0_4 \\0_4 &
    -{\sf c}^l_R\end{array}\right)_{st},
\end{align}
yielding the same condition
\eqref{eq:105}. Hence the twisted first-order condition is satisfied
as~soon~as
\begin{equation}
c^r_R =m,\; d^r_R=n,
\label{eq:8}
\end{equation}
which amounts to
identify $\bb C^r_R$ with $\bb C$. Hence the reduction of $\cal B'$ to
$\cal B$ as defined in \eqref{eq:96}. 
\end{proof}

\bigskip 
One could identify $\bb C^l_R$ with $\bb C$, instead of $\bb
C_R^r$, without changing the result. 
As discussed before definition \ref{deftwistedfirstorder}, one might
also consider a first-order condition where only the commutator with
$D$ is twisted, that is
\begin{equation}
  \label{eq:124}
  [[D_M, A]_\rho, JBJ^{-1}] = 0.
\end{equation}
This is not pertinent in our case however, for this amounts to
permuting ${\bar R}^l_l$ with ${\bar R}_r^r$ in - and only in - the
second term in \eqref{eq:120bis}, which then no longer vanishes as soon
as $R^r_r\neq R_l^l$. 
\bigskip

Proposition \ref{lemmasigma}
deals only with the finite dimensional part of the spectral
  triple. However \eqref{eq:93} is still satisfied with $A, B\in
  C^\infty(\M)\otimes \cal B$ (though, strictly speaking, one
 can no longer talk of 
  ``twisted first-order condition for  $D_M$'', for on
$L^2(\M)\otimes \bb C^{128}$ the operator $D_M$  does not have
a compact resolvent). Proposition \ref{twisted-spec-triple}  is true
for the subalgebra $C^\infty(\M)\otimes\cal B$. Therefore the twisted
first-order condition \eqref{eq:112} with $\cinf\otimes\cal B$ is true for $\Ds + D_M$
since it is true for $\Ds$ and $D_M$ independently. This proves the
first statement of theorem \ref{theo1}.

\section{Twisted-covariant Dirac operators}\setcounter{equation}{0}
\label{section:twistedoperators}

The twisted spectral triple
\begin{equation}
  (\cinf\otimes{\cal B},\,  L^2(\M) \otimes {\mathbb C}^{128},\, \Ds + D_M;\, \rho)
\label{eq:189}
  \end{equation}
of theorem \ref{theo1}
  solves the problem of the non-boundedness of
  the commutators $[\Ds, A]$ raised by the non-trivial action of the
  grand algebra on spinors. But to be of
  interest, this spectral triple should preserve
 the property the grand algebra has been invented for,  that is generating the field $\sigma$
  by a fluctuation of $D_M$, or a twisted version of it. As shown in
  this section this is
  indeed the case, because although
  $\cal B$ is not so grand (it is smaller than $\A_G$), it is neither too small ($\cinf\otimes\cal
  B$ still has non  trivial action on spinors).

\subsection{Twisted fluctuation}
In analogy with gauge fluctuation of almost commutative geometries
described in \S\ref{spectraltripleSM}, we call
\emph{twisted fluctuation of $D$ by $C^\infty(\M)\otimes\cal B$} the substitution of $D = \Ds
+ D_M$ with
   \begin{equation}
    \label{eq:143}
    D_{\bb A} = D + \bb A + J\, \bb A\, J^{-1}
  \end{equation}
where $\bb A$ is twisted $1$-form
\begin{equation}
  \label{eq:125}
  \bb A = B^i [D, A_i]_\rho \quad A_i, B^i\in \cinf\otimes\cal B.
\end{equation}
We do not require $\bb A$ to be selfadjoint, we only ask that
$D_{\bb A}$ is selfadjoint and call it \emph{twisted-covariant Dirac operator}. It is the sum 
 $D_{\bb A} = D_X + D_\sigma$
of the twisted-covariant free Dirac operator
\begin{equation}
  \label{eq:26}
  D_X:= \Ds + \slashed {\mathbb A} + J \slashed{\mathbb A} J^{-1} \quad\quad \slashed{\mathbb A} :=
  B^i[\Ds, A_i]_\rho
\end{equation}
with the twisted-covariant Majorana-Dirac operator
\begin{equation}
  \label{eq:27}
  D_\sigma := D_M + \bb A_M + J \bb A_M J^{-1} \quad\quad \bb A_M:= B^i [D_M, A_i]_\rho.
\end{equation}
In this section, we compute explicitly $D_X$ and $D_\sigma$, and show
that they are parametrized by a vector field $X_\mu$ and a scalar
field $\sig$. 

In the following, 
$A_i=(Q_i,M_i)$ and $B^i=(R^i, N^i$)  are arbitrary elements of $\cinf\otimes\cal
B$, where $i$ a summation index and 
\begin{equation}
Q_i = \left(\begin{array}{cc} Q_{ri}^{r}& 0_4 \\ 0_4 &
    Q_{li}^{l}\end{array}\right)_{st}\,,\quad  M_i = \delta_{\dot s}^{\dot t} \left(\begin{array}{cc} c_{i}^r & 0
      \\ 0 & {\bf \sM}_i\end{array}\right)_{\sI\sJ} 
\end{equation}
with{\footnote{In all this section, the components of the matrices are
  functions on $\M$. To lighten notation we write $M_3(\bb C)$ instead
  of $\cinf\otimes M_3(\bb C)$. The same is true for the various
  copies of $\mathbb H$ and $\mathbb C$.}} ${\bf \sM}_i \in M_3(\bb C)$ and
\begin{equation} 
   Q_{ri}^{r}=\left(
     \begin{array}{cc}
       q^r_{Ri}& 0_2 \\
       0_2 & q^r_{Li}
     \end{array}
  \right)_{\alpha\beta}, \; Q^{l}_{li}=\left(
     \begin{array}{cc}
       q^l_{Ri}& 0_2 \\
       0_2 & q^l_{Li}
     \end{array}
  \right)_{\alpha\beta} 
\label{eq:6}
\end{equation}
with $q_{Li}^l\in \bb H_L^l$, $q_{Li}^r\in \bb H_L^r $ and
\begin{equation}
q_{Ri}^r =\text{diag}\,(c_{i}^r,\bar c_{i}^r), \quad q_{Ri}^l =\text
{diag}\,(c_{i}^l,\bar c_{i}^l) \quad\text{ with }\quad c_i^r\in \mathbb C_R^r,\;\,
c_i^l \in \mathbb C_R^l.
\label{eq:9}
\end{equation}
The components $R^i, N^i$ of $B^i$ are defined similarly, with 
\begin{equation}
d^{ri}\in \mathbb C_R^r,\; d^{li}\in\mathbb C_R^l,\quad
 r_L^{ri}\in\mathbb H_{L}^r , \;  r_L^{ri}\in\mathbb
 H_{L}^l\quad\text{and}\quad
{\bf{\sN}}_i\in M_3(\mathbb C).
\label{eq:11}
\end{equation}
\subsection{Twisted-covariant free Dirac operator $D_X$}
\label{potenziale}
 The twisted fluctuations \eqref{eq:26} of the free Dirac operator
 $\Ds$ in \eqref{eq:3-bis} by 
  $C^\infty(\M)\otimes\cal B$ are parametrized by
 a vector field. 
\begin{prop}
\label{propdx}
 One has
  \begin{equation}
  \label{eq:132}
   D_X = \slashed D + \slashed{\mathbb  X} 
\end{equation}
with
  \begin{equation}
    \label{eq:192}
\slashed{\mathbb X}  :=  -i\gamma^\mu {\mathbb X}_\mu\;,
 \quad\quad
{\mathbb X}_\mu :=\left(
 \begin{array}{cc}
 X_\mu &0_{64} \\ 
  0_{64} & -
    \bar X_\mu
 \end{array}
\right)_{\sC\sD},
  \end{equation}
where we define the  bounded-operator valued vector field {\footnote{To lighten notations we
  omit the parenthesis around $(\del_\mu Q_i)$ and $(\del_\mu \bar M_i)$:
  the latter are bounded operators and act as matrices, not as
  differential operators.}} 
\begin{equation}
  \label{eq:127}
  X_\mu:=\delta^I_J\,\rho(R^i)\,\nabla^S_\mu Q_i - \delta_\alpha^\beta
  \,\bar N^i \nabla^S_\mu \bar  M_i\end{equation}
which commutes with $\gamma^5$ and twisted-commutes with $\gamma^\nu$,
that is for all $\mu,\nu$ one has
\begin{equation}
   \label{eq:76}
 \gamma^\mu X_\nu = \rho(X_\nu)\gamma^\mu,\quad
 \gamma^\mu\rho(X_\nu)=X_\nu\gamma^\mu.
 \end{equation}
\end{prop}
\begin{proof}
Given $A_i=(Q_i, M_i)$ and $B^i=(R^i, N^i)$ in ${\cal B}$, one gets
from (\ref{eq:410}), (\ref{eq:410bis}) and  (\ref{eq:130})
\begin{equation}
\label{eq:410ter}
 \slashed{\bb A}  =-i B^i [\slashed D, A_i]_\rho =-i\left(
\begin{array}{cc}
 \delta^I_J \, \gamma^\mu \rho(R^i)\nabla^S_\mu Q_i&0_{64} \\ 
 0_{64} &\delta_\alpha^\beta\, \gamma^\mu N^i\nabla^S_\mu M_i
\end{array}
\right)_{\partind{CD}}
 \end{equation}
where we used that $N^i$ commutes with $\gamma^\mu$ and $
R^i\gamma^\mu  = \gamma^\mu \rho(R^i)$ (lemma \ref{lemautomorfismo}).    
 By \eqref{eq:12} one gets
\begin{equation}
   \label{eq:126}
    J{\slashed{\bb A}} J^{-1} = -  J{\slashed{\bb A}} J  = i\left(
\begin{array}{cc}
\delta_\alpha^\beta\, \gamma^\mu\bar N^i\nabla^S_\mu \bar M_i & 0_{64}\\
  0_{64} &\delta^I_J \gamma^\mu\rho(\bar R^i)\nabla^S_\mu \bar Q_i
\end{array}
\right)_{\partind{CD}}
 \end{equation}
where we used that $\cal J$
anti-commutes with the $\gamma$'s matrices and commutes with $\nabla_\mu^S${\footnote{$
     \left\{\cal J, \gamma^\mu\right\} = i (\gamma^0 \gamma^2
     {\bar \gamma}^\mu + \gamma^\mu \gamma^0 \gamma^2)cc = 0  
 $
because $\bar{\gamma}^\mu = -\gamma^\mu$ for $\mu =1,3$,
$\bar{\gamma}^\mu =\gamma^\mu$ for $\mu=0,2$. That $\cal J$ commutes
with the spin covariant derivative $\nabla_\mu^S$ is a classical
result, see e.g. \cite[Prop. 4.18]{Walterlivre}.}} so that, inserting ${\cal
J}^2 = -{\bb I}$ before $\nabla_\mu^S$, one obtains
\begin{align}
   \label{eq:141}
   {\cal J}\,(\gamma^\mu N^i \nabla^S_\mu M_i)\,{\cal J} &=\gamma^\mu ({\cal J}
   N^i {\cal J})\, \nabla^S_\mu({\cal J} M_i{\cal J}) = \gamma^\mu\bar
   N^i \nabla^S_\mu \bar M_i,\\
   {\cal J}(\gamma^\mu\rho(R^i)\del_\mu Q_i) {\cal J} &=
   \gamma^\mu ({\cal J}\rho(R^i) {\cal J})\nabla^S_\mu({\cal J} Q_i{\cal J}) =  \gamma^\mu\rho(\bar R^i)\nabla^S_\mu \bar Q_i.
 \end{align}
In both equations above the last term comes from \eqref{eq:31},
noticing that $\rho(R_i)$ and $Q_i$ are now diagonal in the $st$ index
and so commute with $\eta$, while $N_i, M_i$ are proportional to
$\delta_{\dot s}^{\dot t}$, hence commute with $\tau$.
Summing up \eqref{eq:410ter} and \eqref{eq:126}, one obtains 
\begin{equation}
  \label{eq:187}
  \slashed{\mathbb A} + J\slashed{\mathbb A} J^{-1}= - i \gamma^\mu \mathbb X_\mu
\end{equation}
with $\mathbb X_\mu$ as in \eqref{eq:127}. 

$X_\mu$ commuting with
$\gamma^5$ is a consequence of the breaking of $\A_G$ by the
grading condition and can be checked explicitly using \eqref{eq:6} and
\eqref{repa4}. Eq. \eqref{eq:76} follows by direct calculation, writing explicitly
 $X_\mu$ in the $st$ indices
\begin{equation}
  \label{eq:147}
  X_\mu = \delta^{\sI \dot t}_{\sJ\dot s} \left(\begin{array}{cc} 
R^{il}_{\; l}\nabla^r_\mu Q^{\;r}_{ir} & 0_4 \\ 
0_4 &R^{ir}_{\; r}\nabla^l_\mu Q^{\;l}_{il}\end{array}\right)_{st} - \delta_{\alpha s
\dot s}^{\beta t \dot t}\bar N ^i \del_\mu \bar M_i =:  \delta^{\dot t}_{\dot s} \left(\begin{array}{cc}X_\mu^r &
0_{32}\\ 0_{32} &X_\mu^l \end{array}\right)_{st},
\end{equation}
where $\nabla_\mu^{r,l}:= \del_\mu + \omega_\mu^{r,l}$ are defined by the explicit form of the spin connection
\begin{equation}
\omega_\mu =-4
\,\Gamma_{\mu a}^b \gamma^a\gamma_b= \left(\begin{array}{cc}
-4
\,\Gamma_{\mu a}^b\sigma^a \tilde\sigma_b =: \omega^r & 0_2 \\ 0_2
&-4
\,\Gamma_{\mu a}^b\tilde\sigma^a\sigma_b=: \omega^l \end{array}\right)_{st}.
\label{eq:74}
\end{equation}
with $\Gamma_{\mu a}^b$ the Christoffel symbols in the orthonormal
basis.
\end{proof}
   
\begin{lemma} 
\label{lemma-selfadjoint}
$D_X$ is selfadjoint, and called twisted-covariant free Dirac
operator,  if and only if for
  any $\mu= 0, 1,2,3$ one has
  \begin{equation}
 \rho( X_\mu) = - X_\mu^\dagger.
\label{eq:150}
  \end{equation}
\end{lemma}
\begin{proof}
By \eqref{eq:147},
% in the $st$ indices $X_\mu$ is a block diagonal matrix which is
%proportional to $\delta_{\dot s}^{\dot t}$. Explicitly, % ,
% \begin{equation}
%   \label{eq:147}
%   X_\mu = \delta^{\sI \dot t}_{\sJ\dot s} \left(\begin{array}{cc} 
% R^{il}_{\; l}\nabla^S_\mu Q^{\;r}_{ir} & 0_4 \\ 
% 0_4 &R^{ir}_{\; r}\nabla^S_\mu Q^{\;l}_{il}\end{array}\right)_{st} - \delta_{\alpha s
% \dot s}^{\beta t \dot t}\bar N ^i \nabla^S_\mu \bar M_i =:  \delta^{\dot t}_{\dot s} \left(\begin{array}{cc}X_\mu^r &
% 0_{32}\\ 0_{32} &X_\mu^l \end{array}\right)_{st},
% \end{equation}
\begin{equation}
  \label{eq:149}
 \gamma^\mu X_\mu = 
  \left(\begin{array}{cc} 0_{32}  & \sigma^\mu X_\mu^l \\
\tilde \sigma^\mu X_\mu^r & 0_{32}\end{array}\right)_{st} ,\quad 
 (\gamma^\mu X_\mu )^\dagger=  
  \left(\begin{array}{cc} 0_{32}  & \sigma^\mu (X_\mu^r)^\dagger \\
\tilde \sigma^\mu (X_\mu^l)^\dagger& 0_{32}\end{array}\right)_{st} = \gamma^\mu\rho(X_\mu^\dagger),
\end{equation}
where we used that $X_\mu$ commutes with the $\sigma$'s matrices and
$(\sigma^\mu)^\dagger = \tilde \sigma^\mu$. Therefore $\gamma^\mu X_\mu$ is selfadjoint iff 
\begin{equation}
  \label{eq:151}
  \sigma^\mu (X_\mu^r)^\dagger = \sigma^\mu X_\mu^l.
\end{equation}
Since $\text{Tr}\, \bar\sigma^\nu\sigma^\mu = 2\delta^\mu_\nu$ and both $X_\mu^r$ and $X_\mu^l$ are proportional to
$\delta_{\dot s}^{\dot t}$,  the partial trace  on the $\dot s\dot t$
indices of the above equation, where both side have been multiplied by
$\bar\sigma^\lambda$, yields 
$ (X_\mu^r)^\dagger = X_\mu^l$ for any $\mu$, that is
\begin{equation}
  X_\mu^\dagger = \rho(X_\mu).
\end{equation}
The lemma is obtained noticing that by Kato-Rellich theorem $D_X$ is
selfadjoint if and only if $i\gamma^\mu X^\mu$ is selfadjoint, that is
$\gamma^\mu X^\mu$ is
anti-selfadjoint. 
\end{proof}

\subsection{Twisted-covariant Majorana-Dirac operator $D_{\bf
    \sigma}$}

Twisted fluctuations of the Majorana-Dirac operator
$D_M$ are parametrized by a scalar field
$\sig$. To show that, we begin by a short calculation in tensorial notations.
\label{soussectionsigma}
\begin{lemma}
\label{lemmaoneformsingle}
  For  $A=(Q,M)\in {\cal B}$ with components $c^r, c^l\in\mathbb
  C$ as in \eqref{eq:9}, one has
  \begin{equation}
    \label{eq:34}
    [D_M, A]_\rho = \left(\begin{array}{cc}
0_2 & k_R(c^r - c^l) {\cal S} \\ \bar k_R(c^r - c^l)
{\cal S}'&0_2\end{array}\right)_{\sC\sD}\delta_{\dot s}^{\dot t}\,\Xi_{\alpha\sI}^{\beta\sJ}
  \end{equation}
where \begin{equation}
  \label{eq:107}
  {\cal S}= \left(\begin{array}{cc} 1 & 0
      \\ 0&0\end{array}\right)_{st},\quad
  {{\cal S}'} = \left(\begin{array}{cc} 0&0 \\
      0& 1\end{array}\right)_{st}. 
\end{equation}
\end{lemma}
\begin{proof}
Computing explicitly
\eqref{eq:95} with notations \eqref{eq:104} and omitting $k_R$ and
$\bar k_R$ yields
{\small \begin{align}
  \label{eq:102ter}
{\sf D_M} M\! -\! \rho(Q)\sf D_M  & =(\Xi M )_{\dot s \sJ}^{\dot t \sI} (\eta\Xi)_{s\alpha}^{t\beta}  - (\Xi\delta)_{\dot  s\sJ}^{\dot
   t\sI}\,(\rho(Q) \eta\Xi)_{s\alpha}^{t\beta}\\
\nonumber
&= \left(\begin{array}{cc}  
    \sf m  & 0_4 \\ 0_4 &\sf m 
    \end{array}\right)_{\dot s \dot t}\otimes
  \left(\begin{array}{cc}
     \Xi_\sJ^\sI & 0_4 \\0_4& -\Xi_\sJ^\sI
   \end{array}\right)_{st} - 
 \left(\begin{array}{cc}  \Xi_\alpha^\beta& 0_4 \\ 0_4 &\Xi_\alpha^\beta\end{array}\right)_{\dot
   s \dot t}\otimes
 \left(\begin{array}{cc}
     {\sf c}^l_R& 0_4 \\0_4 &-{\sf c}^r_R\end{array}\right)_{st} \\
\nonumber &=\left(\begin{array}{cc} 
    \left(\begin{array}{cc}
         (m- c^l_R)\,\Xi_{\alpha\sI}^{\beta\sJ}& 0 \\ 0
     & (m-c^l_R)\,\Xi_{\alpha\sI}^{\beta\sJ}
   \end{array}\right)_{\dot s\dot t} & 0_{32}\\
0_{32}&  \left(\begin{array}{cc}
         - (m- c^r_R)\,\Xi_{\alpha\sI}^{\beta\sJ}& 0 \\ 0
     & - (m-c^r_R)\,\Xi_{\alpha\sI}^{\beta\sJ}
   \end{array}\right)_{\dot s\dot t} \end{array}\right)_{s t}\\[5pt]
\label{eq:102quar}
{\sf D_M}^\dagger Q \!-\! M{\sf D}_\nu^\dagger&=  (\Xi\delta)_{\dot s\sJ}^{\dot t\sI}  \,(\eta\Xi Q)_{s\alpha}^{t\beta}  -
(\Xi M)_ {\dot s\sJ}^{\dot t\sI}(\eta\Xi)_{s\alpha}^{t\beta}\\
\nonumber &=
 \left(\begin{array}{cc} \Xi_\alpha^\beta& 0_4 \\0_4 & \Xi_\alpha^\beta  \end{array}\right)_{\dot s \dot t}\otimes\left(\begin{array}{cc}
     \bar{\sf c}^r_R& 0_4 \\0_4 &-\bar k_R {\sf c}^l_R\end{array}\right)_{st} - \left(\begin{array}{cc} \sf m& 0_4 \\0_4 & \sf m
    \end{array}\right)_{\dot s \dot t}\otimes\left(\begin{array}{cc}
  \bar  \Xi_\sJ^\sI& 0_4 \\0_4 &-\bar \Xi_\sJ^\sI\end{array}\right)_{st} \\[5pt]
\nonumber &= \left(\begin{array}{cc} 
    \left(\begin{array}{cc}
       \bar (c^r_R -m)\,\Xi_{\alpha\sI}^{\beta\sJ}& 0 \\ 0
     & \bar R (c^r_R - m)\,\Xi_{\alpha\sI}^{\beta\sJ}
   \end{array}\right)_{\dot s\dot t} & 0_{32}\\
0_{32}&  \left(\begin{array}{cc}
        -\bar (c^l_R-m)\,\Xi_{\alpha\sI}^{\beta\sJ}& 0 \\ 0
     & -\bar (c^l_R-m)\,\Xi_{\alpha\sI}^{\beta\sJ}
   \end{array}\right)_{\dot s\dot t} \end{array}\right)_{s t}.
\end{align}}
Identifying $c_R^r$ with $m$ following \eqref{eq:8} yields the result, where we
drop the
index $R$ to match notation \eqref{eq:9}.
\end{proof}

\begin{prop}
\label{propsigma} The selfadjoint twisted fluctuation \eqref{eq:27} of the Majorana-Dirac operator
$D_M = \gamma^5D_R$ by $\cinf\otimes\cal B$, called twisted-covariant
Majorana-Dirac operator,  is 
\begin{equation}
  \label{eq:14}
  D_\sigma = {\boldsymbol \sigma}\gamma^5 D_R
\end{equation}
where
\begin{equation}
{\boldsymbol \sigma} = (\mathbb I +
\gamma^5\phi)
\label{eq:17}
\end{equation}
 with $\phi$ a real scalar field.
\end{prop}
\begin{proof}
Let $B^i=(R^i, N^i)$ as in \eqref{eq:11}.  From lemma \ref{lemmaoneformsingle} one gets
  \begin{equation}
    \label{eq:106}
  {\mathbb A}_M=  B^i [ D_M, A_i]_\rho = \, \phi\, \left(\begin{array}{cc}
      0_{2} & k_R\, {\cal S}\\ 
        \bar k_R \, {\cal S}' & 0_{2}\end{array}\right)_{\sC\sD} \delta_{\dot s}^{\dot t}\;\Xi_{\sI\alpha}^{\sJ\beta}
\end{equation}
where
\begin{equation}
  \label{eq:13}
  \phi := d^{ir}(c^r_i-c^l_i). 
\end{equation}
One has ${\cal  J}({\cal S}\delta_{\dot s}^{\dot t}){\cal  J}=
    -{\cal S}\delta_{\dot s}^{\dot t}$ and ${\cal  J}({\cal S'}\delta_{\dot s}^{\dot t}){\cal  J}= -{\cal S}'\delta_{\dot s}^{\dot t}$.
Hence
\begin{equation}
  \label{eq:108}
  J \mathbb A_M  J^{-1} =-  J \mathbb A_M  J =
  \bar \phi\,\left(\begin{array}{cc} 0_{2} & k_R {\cal S}'\\ \bar k_R {\cal S} & 0_{2}\end{array}\right)_{\sC\sD}\delta_{\dot s}^{\dot t}\;\Xi_{\sI\alpha}^{\sJ\beta}
\end{equation}
so that
\begin{equation}
  \label{eq:109}
 D_M +  \mathbb A_M + J \mathbb A_M J^{-1} =
  \left(\begin{array}{cc} 0_{2} & k_R\,(\eta_s^t +\phi\,{\cal S} + \bar\phi\,{\cal S}')\\
      \bar k_R\,(\eta_s^t +\phi\,{\cal S}' +\bar\phi {\cal S})&0_2\end{array}\right)_{\sC\sD}\delta_{\dot s}^{\dot t}\;\Xi_{\sI\alpha}^{\sJ\beta}.
\end{equation}
 It is selfadjoint if and only if $\phi=\bar \phi$.
Then 
\begin{align}
  \label{eq:4}
  D_\sigma &:=  D_M +  \mathbb A_M + J \mathbb A_M J^{-1}
  =\left(\begin{array}{cc}
0_4 & k_R(\gamma^5 + \phi \mathbb{I}_4)\\
\bar k_R(\gamma^5 + \phi \mathbb{I}_4)&
0_4\end{array}\right)_{\sC\sD}\;\Xi_{\sI\alpha}^{\sJ\beta},\\
& \; = (\gamma^5 + \phi \mathbb{I}) D_R.
\end{align}
Factorizing by $\gamma^5$, one gets the result.
\end{proof}
\bigskip

Propositions \ref{propdx} and \ref{propsigma} prove the second
statement of theorem \ref{theo1}. The field $\boldsymbol \sigma$ in \eqref{eq:17} is slightly different from the one obtained in
\cite{Devastato:2013fk} by a non-twisted fluctuation of
$D_M$ by $\A_{sm}\otimes~C^\infty(\M)$, namely 
\begin{equation}
\sigma= (1+\phi)\mathbb I.
\label{eq:25}
\end{equation}
We comment on that in the conclusion.

\section{Breaking of the grand symmetry to the standard model}\setcounter{equation}{0}
\label{section:breaking}

We prove the third and fourth point of theorem \ref{theo1} by computing the
spectral action for the twisted-covariant
Dirac operator
\begin{equation}
D_{\bb A} = D_X + D_\sig,
\label{eq:198}
\end{equation}
where $D_X$ and $D_\sig$
have been obtained by twisted fluctuation of $\Ds$ and $D_M$ in
\eqref{eq:132} and \eqref{eq:14}. More precisely, we show that the
potential part of this action is minimum when the Dirac operator $\slashed D +
D_M$ of
the twisted spectral triple is fluctuated by a subalgebra of $ C^\infty(\M)\otimes {\cal B} $
which is invariant under the automorphism $\rho$. The maximal
such sub-algebra is precisely  the algebra $\cinf \otimes \A_{sm}$ of the standard
  model. Indeed  by \eqref{eq:91} an element $(Q,M)$ of $\cal B$ is invariant by
the automorphism $\rho$ if and only if
\begin{equation}
\rho(Q) = Q,
\label{eq:129}
\end{equation}
which means $\bb H_R^r = \bb H_R^l$ and $\bb C_L^r = \bb C_L^l$,
that is $(Q, M)\in \A_{sm}$.

We begin by some recalls on the spectral action, then we establish
the generalized Lichnerowicz formula for $D_{\bb A}$ and finally we
study the potential for the vector field, the scalar field, and their interaction.

\subsection{Spectral action}
\label{subsec:spectral}

A striking application of noncommutative geometry to physics is to give a gravitational
interpretation of the standard model \cite{Connes:1996fu}. By this,
one intends that the bosonic part of the SM
Lagrangian is deduced from an action which is purely geometric, that
is which depends only on the spectrum of the covariant Dirac
operator $D_A$ \eqref{eq:30} of the almost commutative geometry
of the standard model. The most obvious way to define such an action consists in
counting the eigenvalues lower than a given energy
scale $\Lambda$. This is the spectral action \cite{Chamseddine:1996kx} 
\begin{equation}
S=\Tr f\left(\frac{D_A^2}{\Lambda^2}\right) 
\label{eq:spectral_action}
\end{equation}
where $f$ is a positive cutoff function, usually the (smoothened)
characteristic function on the interval $[0,1]$.
It has an asymptotic expansion in power series~of~$\Lambda$, 
\begin{equation}
\label{asympexp}
\sum_{n\geq0} f_{4-n}\, \Lambda^{4-n}\, a_n(D_A^2/\Lambda^2)
\end{equation}
where the $f_n$ are the momenta of $f$
and  the $a_n$ the Seeley-de Witt coefficients
which are nonzero only for $n$ even.
To compute these coefficients, one usually starts with $D_A^2$
  written as an elliptic operator of Laplacian type, 
\begin{equation}
\label{laplacian}
D_A^2=-(g^{\mu\nu}\,\nabla_\mu^S\,\nabla_\nu^S +\alpha^\mu\,\nabla^S_\mu+\beta),
\end{equation}
and introduces the covariant derivative
\begin{equation}
  \label{eq:155}
  \nabla_\mu := \nabla^S_\mu +\omega_\mu 
\end{equation}
associated with the connection $1$-form
\begin{equation}
  \label{eq:156}
  \omega_\mu:=\frac12 g_{\mu\nu}\left(\alpha^\nu+g^{\sigma\rho}
  \Gamma^\nu_{\sigma\rho}\right). 
\end{equation}
This yields the generalized Lichnerowicz formula
\begin{equation}
  \label{eq:157}
  D_A^2 = -\nabla_\mu \nabla^\mu -E
\end{equation}
where
\begin{align}
\label{laplawithoutdilaton} E:=\beta-g^{\mu\nu}\left(\nabla^S_\mu\, \omega_\nu+\omega_\mu\omega_\nu-\Gamma^\rho_{\mu\nu}\omega_\rho\right).
\end{align}
The coefficients
$a_n$ are then computed by usual technics of heat kernel. The first ones are \cite{Gilkey1984,Vassilevich:2003fk}
\begin{eqnarray}
a_0&=&\frac 1{16\pi^2}\int\dd x^4 \sqrt{g}
\Tr(Id), \hfill\\ 
a_2&=&\frac 1{16\pi^2}\int\dd x^4 \sqrt{g}
\Tr\left(-\frac R6+E\right)\nonumber \hfill\\
a_4&=&\frac{1}{16\pi^2}\frac{1}{360}\int\dd x^4 \sqrt{g} \Tr(-12\nabla^\mu\nabla_\mu R +5R^2-2R_{\mu\nu}R^{\mu\nu}\\
& &+2R_{\mu\nu\sigma\rho}R^{\mu\nu\sigma\rho}-60RE+180E^2+60\nabla^\mu\nabla_\mu
E+30\Omega_{\mu\nu}\Omega^{\mu\nu}) \label{spectralcoeff}
\end{eqnarray}
where $R_{\mu\nu}$ is the Ricci tensor, $-R$ the scalar curvature and
$\Omega_{\mu\nu}$ the curvature of the connection $\omega_\mu$.
Applied to the spectral triple \eqref{eq:02} of the standard model,
fluctuated according to \eqref{eq:30}, the expansion
(\ref{asympexp}) yields the bosonic part of 
Lagrangian of the standard model - including the Higgs - minimally coupled with
gravity~\cite[Sect.~4.1]{Chamseddine:2007oz}. For the fermionic action
and how it is related to the spectral action see \cite{Adrianov:2011fk}, \cite{Lizzi:2010vn} and for a complete and pedagogical treatment of
the subject, see the recent book \cite{Walterlivre}. 

Here we compute the asymptotic expansion \eqref{asympexp} for the twisted covariant Dirac operator $D_{\mathbb
  A}$ \eqref{eq:198}.
For simplicity we restrict to the flat
case $g^{\mu\nu}=\delta^{\mu\nu}$, so that \eqref{eq:155}, \eqref{eq:156} and \eqref{laplawithoutdilaton} reduce to
\begin{equation}
  \label{eq:-38}
\nabla_{\mu}=\partial_{\mu}+\omega_{\mu}\,,\,\quad\omega_{\mu}=\frac{1}{2}g_{\mu\nu}\alpha^{\nu}\,,\,\quad
E=\beta-g^{\mu\nu}\left(\partial_{\mu}\omega_{\nu}+\omega_{\mu}\omega_{\nu}\right),
\end{equation}
that is
\begin{equation}
  \label{eq:38}
 \nabla_\mu = \del_\mu + \frac 12 \alpha_\mu, \quad  E = \beta - \frac 14 \alpha \cdot \alpha - \frac 12
\del_\mu \alpha^\mu
\end{equation}
where $\alpha\cdot\alpha := g_{\mu\nu}\alpha^\mu \alpha^\nu$ denotes
the inner product defined by the Riemannian metric. Furthermore, in all this section, we consider fluctuations such that $D_X$
and $D_\sig$ are selfadjoint, meaning $X_\mu$ satisfies lemma
\ref{lemma-selfadjoint} and $\phi$ is a real field.

\subsection{Lichnerowicz formula for the twisted-covariant Dirac operator}
\label{section:squaretwisteddirac}

We define
\begin{equation}
\slashed X :=-i\gamma^\mu X_\mu, \quad\slashed \rho(X) := -i\gamma^\mu
\rho(X_\mu).
\label{eq:5}
\end{equation}
These are selfadjoint operators since by \eqref{eq:76} and lemma \ref{lemma-selfadjoint} one has
\begin{equation}
  \label{eq:75}
  \slashed X^\dagger = i \,X_\mu^\dagger \gamma^\mu = -i
  \rho(X_\mu)\gamma^\mu = -i \gamma^\mu X_\mu =  \slashed X,
\end{equation}
and similarly for $\slashed\rho(X)$. 
The same is true for
\begin{equation}
  \label{eq:85}
  \slashed{\bar X} := -i\gamma^\mu\bar X_\mu,\quad   \slashed\rho({\bar X}) := -i\gamma^\mu\rho(\bar X_\mu).
\end{equation}
Similar equations hold for the field $\sig$, by extending the automorphism $\rho$ to $\cal B(\cal H)$
as the conjugate action of the unitary operator that exchanges the
indices $l$ and $r$ in the basis of $\sf H$. Doing so, one gets
$\rho(\gamma^5) = - \gamma^5$, that is
\begin{equation}
  \label{eq:66}
  \rho(\sig) = \mathbb I - \gamma^5 \phi.
\end{equation} 
Thus $\sig$ twisted-commutes with $\gamma^\mu$ - as $X_\mu$ in
\eqref{eq:76} - for the anti-commutativity of $\gamma^\mu$ and $\gamma^5$ yields
\begin{equation}
\label{eq:81}
\gamma^\mu \sig = \rho(\sig)\gamma^\mu,\quad\quad \gamma^{\mu}\rho(\sig) = \sig\gamma^\mu.
\end{equation}

The standard model algebra $\A_{sm}$ is the subalgebra of $\cal B$
invariant under the twist. To measure how far the grand symmetry
is from the SM, we introduce as physical degrees of freedom
the fields
\begin{equation}
  \label{eq:171}
  \Delta(X)_\mu:= X_\mu -\rho(X_\mu), \quad \Delta(\sig) := (\sig - \rho(\sig))D_R.
\end{equation}
Both are selfadjoint, $\Delta(X)_\mu$ by  lemma
\ref{lemma-selfadjoint}, $\Delta(\sig)$ because $\sig$ and $D_R$ are
selfadjoint and commute. Moreover, by \eqref{eq:76} and \eqref{eq:81} one has
\begin{equation}
  \label{eq:146}
  \left\{\gamma^\mu, \Delta(X)_\nu\right\}
  =\left\{\gamma^\mu,\Delta(\sig)\right\} = 0, 
\end{equation} 
while  $\gamma^5$ commuting with $X_\mu$ and $\sig$ guarantee that
\begin{equation}
\left[\gamma^5, \Delta(X)_\nu\right] =\left[\gamma^5,
    \Delta(\sig)\right] = 0. 
\end{equation}

We write
\begin{equation}
  \label{eq:166}
  \rho(\mathbb X_\mu) := \left(
\begin{array}{cc} 
  \slashed\rho(X) & 0_{64} \\
  0_{64}&-\slashed\rho(\bar X)
\end{array}\right)_{\sC\sD},\quad\quad
  \Delta(\XX)_\mu := \XX_\mu - \rho(\XX_\mu),
\end{equation}
and in agreement with  \eqref{eq:132} and \eqref{eq:192} written as
\begin{equation}
  \label{eq:1999}
  \slashed\XX = -i \gamma^\mu {\mathbb X}_\mu = \left(
    \begin{array}{cc}
      \slashed X & 0_{64} \\
      0_{64}&-\bar{\slashed X}\end{array}\right)_{\sC\sD},
\end{equation}
we also define the selfadjoint operators
\begin{equation}
  \label{eq:199}
  \slashed\rho(\XX) :=-i\gamma^\mu \rho(\XX_\mu), \quad \slashed\Delta(\XX) := \slashed\XX - \slashed\rho(\XX).
\end{equation}
Finally,  we let
\begin{equation}
  \label{eq:177}
  D_\mu := \del_\mu + \text{ad}\; \XX_\mu
\end{equation}
denote the covariant derivative associated with the connection $\XX_\mu$.

\begin{prop} The square of the twisted-covariant Dirac operator
\eqref{eq:198} is
  \begin{equation}
    \label{eq:47}
    D_{\bb A}^2 = -\left(\gamma^\mu\gamma^\nu\del_\mu\del_\nu + (\alpha_X^\mu+\alpha^\mu_{\sig}) \del_\mu
+\beta_X +\beta_{X\sig}+ \beta_{\sig}\right)
  \end{equation}
where
 \begin{eqnarray}\label{eq:alpha}
\alpha^{\mu}_X & := i\left\{\slashed{\mathbb X}, \gamma^\mu\right\} ,
\quad 
\beta_X = i\gamma^\mu (\del_\mu\slashed{\mathbb X}) -  \slashed{\mathbb X}\slashed{\mathbb X},
\end{eqnarray} 
while
\begin{equation}
  \label{eq:48}
\alpha^\mu_{\sig}:=  i
   \gamma^{\mu}\gamma^{5} \Delta(\sig),\quad\quad   \beta_{\sig} := -\sig^2 D_R^2,
\end{equation}
and
\begin{equation}
  \label{eq:153}
\beta_{X\sig}:= i\gamma^\mu\gamma^5\left(D_\mu (\sig D_R) + \Delta(\sig)\,\XX_\mu\right).
\end{equation}
\end{prop}\begin{proof}
One has
$  D_{\bb A}^2 =D_X^2 + D_\sigma^2 +  \left\{D_X, D_\sigma\right\}.$
 By \eqref{eq:132}, the first term is
\begin{align}
D_X^{2}&=-\gamma^\mu(\del_\mu + \XX_\mu)\, \gamma^\nu(\del_\nu + \XX_\nu) \\
 &= -\gamma^\mu\gamma^\nu \del_\mu\del_\nu -i\left\{\slashed
   \XX,\gamma^\mu\right\}\del_\mu -i \gamma^\mu (\del_\mu \slashed\XX)
 +\slashed \XX\slashed \XX\\
&= -\left(\gamma^\mu\gamma^\nu\del_\mu\del_\nu
   +\alpha^{\mu}_X\partial_{\mu}+\beta_X\right).
\label{eq:161}
\end{align}
By propositions \ref{propdx} one has
\begin{align}
  \label{eq:193}
  \left\{D_X, D_\sig\right\} = -i\left\{\gamma^\mu \del_\mu,
    D_\sig\right\} -i \left\{\gamma^\mu \XX_\mu,
    D_\sig\right\}. 
\end{align}
From proposition \ref{propsigma}, using \eqref{eq:81} and $\left\{\gamma^5,
  \gamma^\mu\right\} = [\gamma^5, \sig] = 0$, one gets
\begin{align}
  \label{eq:194}
  \left\{\gamma^\mu \del_\mu, D_\sig\right\} =   \left\{\gamma^\mu\del_\mu, \gamma^5\sig D_R\right\}
&=\gamma^\mu\gamma^5 \del_\mu \sig  D_R  - \gamma^\mu\gamma^5\rho(\sig) D_R \,\del_\mu,\\
& =\gamma^\mu\gamma^5 \left(\del_\mu\sig D_R\right)+ \gamma^\mu \gamma^5\Delta(\sig) D_R\,\del_\mu.
\label{eq:194-ter}\end{align}
Similarly, using that $\gamma^5$ commutes with $X_\mu$, hence with
$\mathbb X_\mu$, one has
\begin{align}
  \label{eq:195}
  \left\{\gamma^\mu \XX_\mu,
    D_\sig\right\} &=  \left\{\gamma^\mu \XX_\mu,
    \sig\gamma^5 D_R\right\}
 = \gamma^\mu \XX_\mu
    \sig\gamma^5 D_R + \sig\gamma^5 D_R\gamma^\mu
    \XX_\mu,\\
&= \gamma^\mu\gamma^5 [\XX_\mu, \sig  D_R]_\rho
= \gamma^\mu\gamma^5\left( [\XX_\mu, \sig  D_R] +
  \Delta(\sig)\,\XX_\mu\right).
\label{eq:195-ter}
\end{align}
Summing \eqref{eq:195-ter} and \eqref{eq:194-ter}, and using the
definition \eqref{eq:177} of $D_\mu$, one rewrites \eqref{eq:193} as 
$\left\{D_X, D_\sig\right\} = -(\alpha^\mu_{\sig}\del_\mu +
\beta_{X\sig})$.
Finally from \eqref{eq:14} one has $D_\sig^2 =-\beta_\sig$.
\end{proof}
Remarkably the contributions $\alpha_{\sig}^\mu$  of the
anti-commutator of $D_X$ and $D_\sig$ to the order one part of
$D_{\mathbb A}^2$ depends on $\sig$ only, and not on $X$. The
same is true for $\beta_\sig$. The contributions
$\alpha_X^\mu$ and $\beta_X$ of $D_X$ depend on $X$ only, and not on
$\sig$. Thus in the Lichnerowicz formula for
$D_{\mathbb A}^2$, that is
\begin{equation}
  \label{eq:173}
  D_{\mathbb A}^2 = -\nabla_\mu\nabla^\mu -E
\end{equation}
with
\begin{equation}
  \label{eq:174}
  \nabla_\mu = \del_\mu + \frac 12g_{\mu\nu} (\alpha_X^\nu+\alpha_{\sig}^\nu),
\end{equation}
the bounded endormorphism $E$ is the sum 
 \begin{equation}
   \label{eq:70}
   E = E_X + E_\sig + E_{X\sig}
 \end{equation}
 of three terms:
 \begin{equation}
 E_X := \beta_X - \frac 14 \alpha_X \cdot \alpha_X - \frac 12 \del_\mu \alpha_X^\mu,
 \label{eq:73}
 \end{equation}
which depends only on $X$,
\begin{equation}
\label{eq:73-1}
E_\sig := \beta_\sig - \frac 14\alpha_{\sig}\cdot \alpha_{\sig} -
\frac 12 \del_\mu \alpha_{\sig},
  \end{equation}
that depends only on $\sig$,  and an interaction term
\begin{equation}
  \label{eq:72}
  E_{X\sig}:= \beta_{X\sig} -  \frac 14\left(\alpha_X\cdot
    \alpha_{\sig} + \alpha_{\sig}\cdot \alpha_X\right). 
\end{equation}

\subsection{Deviation from the non-twisted case}
\label{deviation}

We write the endomorphisms $E_X, E_\sig$ and
$E_{X\sig}$ that appear in the Lichnerowicz formula \eqref{eq:173} of the
twisted-covariant Dirac operator $D_{\mathbb A}$ in
terms of the physical degrees of freedom $\Delta(\sig)$,
$\Delta(X)_\mu$ defined in \eqref{eq:171}. This will permit to measure how far the twisted
spectral triple \eqref{eq:189} is from the spectral
triple of the SM, basing our measure on the spectral
action. Let us start with a technical lemma.\begin{lemma} 
\label{lemmabreaking}
One has
  \begin{align}
\label{alphaxalphax}
&\alpha_X\cdot\alpha_X = 2\left\{\slashed\rho(\XX), \slashed\XX\right\} - 4 \slashed \XX \slashed\XX  
+ 4\XX\cdot\rho(\XX),\\
\label{alphaxalphasig}
&\alpha_X\cdot\alpha_{\sig}+  \alpha_{\sig}\cdot\alpha_{X} = -2i\gamma^\mu \gamma^5\left\{2\XX_\mu  -\Delta(\XX)_\mu, \Delta(\sig) \right\}.
\end{align}
\end{lemma}
\begin{proof} One has
  \begin{align}
  \label{eq:53}
\alpha^\mu_X\cdot \alpha^\nu_X  &= -
  (\slashed\XX \gamma^\mu + \gamma^\mu\slashed\XX)   (\slashed\XX
  \gamma^\nu + \gamma^\nu\slashed\XX) \\
\label{eq:53bis}
&=-\left(
   \slashed \XX \gamma^\mu \slashed\XX\gamma^\nu+\slashed\XX \gamma^\mu  \gamma^\nu\slashed\XX+   \gamma^\mu\slashed\XX\slashed\XX
  \gamma^\nu +   \gamma^\mu\slashed\XX\gamma^\nu\slashed\XX\right).
\end{align}
The contraction of the second term with the metric is easily computed using $g_{\mu\nu}\gamma^\mu\gamma^\nu = 4 \mathbb I$: 
\begin{equation}
  \label{eq:203}
 g_{\mu\nu} \, \slashed \XX \gamma^\mu  \gamma^\nu\slashed \XX = 4\slashed
\XX\slashed \XX. 
\end{equation}
For the remaining terms, \eqref{eq:76} written as
\begin{equation}
\XX_\mu \gamma^\nu = \gamma^\nu \rho(\XX_\mu)
\label{eq:154}
\end{equation}
together with
$\gamma^\mu\gamma^\nu = 2g^{\mu\nu}{\mathbb I} - \gamma^\nu\gamma^\mu$ yields
\begin{align}
  \label{eq:40}
  \slashed \XX \gamma^\mu = -2i\rho(\XX^\mu) -\gamma^\mu \slashed\rho(\XX),
 \quad\gamma^\mu \slashed \XX = -2i\XX^\mu- \slashed\rho(\XX)\gamma^\mu.
 \end{align}
Therefore
\begin{align}
  \label{eq:67}\nonumber
g_{\mu\nu}  \,\slashed \XX \gamma^\mu \slashed \XX\gamma^\nu &=-
g_{\mu\nu} \, \slashed \XX \gamma^\mu  \left(  2i\rho(\XX^\nu) + \gamma^\nu
\slashed\rho(\XX) \right) =   - 2\slashed \XX\, \slashed\rho(\XX);\\[5pt] 
g_{\mu\nu}\gamma^\mu\slashed \XX \gamma^\nu\slashed \XX &=
-g_{\mu\nu}\left(2i\XX^\mu +
  \slashed\rho(\XX)\gamma^\mu\right) \gamma^\nu\slashed \XX = - 2 \slashed\rho(\XX)\slashed \XX;
\\[5pt]
\nonumber
g_{\mu\nu}\gamma^\mu\slashed \XX\slashed \XX
  \gamma^\nu  &= g_{\mu\nu}\left( 2i\XX^\mu+
    \slashed\rho(\XX)\gamma^\mu\right)\left(2i\rho(\XX^\nu) + \gamma^\nu
    \slashed\rho(\XX)\right),\\
\nonumber
&=-4 \XX\cdot \rho(\XX) - 2 \slashed\rho(\XX)\slashed\rho(\XX) -  2
\slashed\rho(\XX)\slashed\rho(\XX) + 4\slashed\rho(\XX)
\slashed \rho(\XX) =  - 4 \XX\cdot \rho(\XX).
\end{align}
Hence the contraction of \eqref{eq:53bis} by $g_{\mu\nu}$ yields \eqref{alphaxalphax}.

To obtain \eqref{alphaxalphasig}, one starts with \eqref{eq:alpha} and
\eqref{eq:48} together with \eqref{eq:146}. This gives
\begin{align}
  \label{eq:88}
 \alpha^{\mu}_{X} \alpha^{\nu}_{X\sig} +\alpha^{\mu}_{X\sig}
    \alpha^{\nu}_X &:= -\left\{\slashed \XX, \gamma^{\mu}\right\}
          \gamma^\nu\gamma^5\Delta(\sig) -
          \gamma^\nu\gamma^5\Delta(\sig)
          \left\{\slashed{\XX},\gamma^{\mu}\right\},\\
  \label{eq:88-bis}
&= -\left\{\slashed \XX, \gamma^{\mu}\right\}
          \gamma^\nu\gamma^5\Delta(\sig)  -
          \gamma^5\Delta(\sig) 
          \gamma^\nu\left\{\slashed{\XX},\gamma^{\mu}\right\}.
\end{align}
By \eqref{eq:40} one has  
\begin{align}
  \label{eq:201}
   g_{\mu\nu} \left\{\slashed\XX, \gamma^{\mu}\right\} \gamma^\nu&=
     - g_{\mu\nu} \left(2i\rho(\XX^\mu)  +\gamma^\mu
     \slashed\rho(\XX) + 2i\XX^\mu +
 \slashed\rho(\XX)\gamma^\mu\right)\gamma^\nu,\\ &=
     2\slashed\XX  - 4\slashed\XX + 2\slashed\rho(\XX) -
 4\slashed\rho(\XX) = -4 \slashed\XX
  + 2\slashed{\Delta}(\XX),
\end{align}
and similarly $g_{\mu\nu} \,\gamma^\mu\left\{\slashed{\XX}, \gamma^{\nu}\right\}
   =   -4 \slashed\XX
  + 2\slashed{\Delta}(\XX)$. 
Therefore \eqref{eq:88-bis} gives
  \begin{align}
\label{eq:86}
    \alpha_{X}\cdot \alpha_{\sig} +\alpha_{\sig}\cdot\alpha_X =
    2\left\{2\slashed\XX  - \slashed \Delta(\XX), \gamma^5\Delta(\sig)
    \right\} &=
   2 \left\{-i\gamma^\mu(2\XX_\mu  - \Delta(\XX)_\mu),
     \gamma^5\Delta(\sig) \right\} ,\\ \nonumber
&= -2i\gamma^\mu \gamma^5\left\{2\XX_\mu  -\Delta(\XX)_\mu, \Delta(\sig) \right\}
 \end{align}
where in the last line we use that  $\gamma^\mu$ anticommutes with both
$\gamma^5$ and $\Delta(\sig)$, while $\gamma^5$ commutes with both
$\XX_\mu$ and $\Delta(\XX)_\mu$.
 \end{proof}

\begin{prop}
\label{Lichnedetail} One has
\begin{align}
  \label{eq:178}
  E_X &= \frac 12 \gamma^\mu\gamma^\nu\left(\mathbb{F}_{\mu\nu}
    +D_\nu\,\Delta(\XX_\mu) +
    \Delta(\XX)_\mu\Delta(\XX)_\nu\right),\\
\label{eq:178-1}
E_{X\sig}& = i\gamma^\mu\gamma^5\left( D_\mu (\sig D_R) - \frac 12[\XX_\mu, \Delta(\sig)]  +  \frac 12\left\{3\XX_\mu - \Delta(\XX)_\mu, \Delta(\sig) 
  \right\}\right)\\
\label{eq:178-2}
E_{\sig} &= \Delta(\sig)^2- \sig^2 D_R^2 - \frac
i2\gamma^\mu\gamma^5\del_\mu\Delta(\sig).
\end{align}
where
\begin{equation}
  \label{eq:176}
  {\mathbb F}_{\mu\nu}:= (\del_\mu \XX_\nu) - (\del_\nu \XX_\mu)
  +[\XX_\mu, \XX_\nu]  
\end{equation}
is the field strength of $\XX_\mu$
\end{prop}
\begin{proof}
By \eqref{eq:73}, \eqref{eq:alpha} and lemma \ref{lemmabreaking} one gets
\begin{equation}
  \label{eq:169}
  E_X =\frac i2\left[\gamma^\mu, (\del_\mu \slashed \XX)\right]  - 
    \frac 12 \left\{\slashed \rho(\XX), \slashed \XX\right\} - \XX\cdot \rho(\XX).
\end{equation}
One further computes, writing $\Delta_\mu$ for $\Delta(\XX)_\mu$,
\begin{align}
\nonumber
   -\frac 12 \left\{\slashed \rho(\XX), \slashed \XX\right\}- \XX\cdot
   \rho(\XX) &= \frac 12\left(\gamma^\mu\rho\left(\,\XX_\mu \right)\gamma^\nu \XX_\nu + \gamma^\mu \XX_\mu
   \gamma^\nu\rho\left(\XX_\nu \right) - \left(\gamma^\mu\gamma^\nu +
   \gamma^\nu\gamma^\mu \right) \XX_\mu\rho\left(\XX_\nu \right)\,\right),\\
\nonumber &=  \frac 12\gamma^\mu\gamma^\nu\left(\,\XX_\mu \Delta_\nu +
  \rho(\XX_\mu)\rho(\XX_\nu)  - \XX_\nu\rho(\XX_\mu)\, \right),\\
\nonumber &=  \frac 12\gamma^\mu\gamma^\nu\left(\XX_\mu \Delta_\nu +
  (\XX_\mu - \Delta_\mu)(\XX_\nu -\Delta_\nu)  - \XX_\nu(\XX_\mu-\Delta_\mu)
\right),\\
\label{eq:179}&=  \frac 12\gamma^\mu\gamma^\nu\left([\XX_\mu, \XX_\nu] +
  \Delta_\mu\Delta_\nu  + [\XX_\nu,\Delta_\mu]
\right).
\end{align}
As well,
\begin{align}
\nonumber
  \frac i2 \left[\gamma^\mu, (\del_\mu \slashed \XX)\right]  &= \gamma^\mu (\del_\mu \gamma^\nu
   \XX_\nu) - (\del_\mu \gamma^\nu \XX_\nu) \gamma^\mu,\\ 
\label{eq:180}
&=
   \gamma^\mu\gamma^\nu\left( \del_\mu X_\nu - \del_\nu
     \rho(X_\mu)\right)=   \gamma^\mu\gamma^\nu\left( \del_\mu X_\nu -  \del_\nu X_\mu +
  \del_\nu \Delta_\mu\right).
  \end{align} 
The sum of \eqref{eq:180} and \eqref{eq:179} gives \eqref{eq:178}.
\smallskip

From \eqref{eq:72}, \eqref{eq:48} and
lemma \ref{lemmabreaking} one obtains
\begin{align}
  \label{eq:42}
  E_{X\sig} &= i\gamma^\mu\gamma^5\left(D_\mu (\sig D_R) +
    \Delta(\sig)\,\XX_\mu\right) + i\gamma^\mu
  \gamma^5\left\{\XX_\mu  -\frac 12\Delta(\XX)_\mu, \Delta(\sig) 
  \right\}.
\end{align}
Eq. \eqref{eq:178-1} then follows writing 
\begin{align}
  \label{eq:33}
   \Delta(\sig) \XX_\mu = \frac 12\left\{ \XX_\mu , \Delta(\sig)\right\} -\frac12 \left[
     \XX_\mu, \Delta(\sig) \right]
\end{align}

To prove \eqref{eq:178-2} one uses $\left\{\Delta(\sig), \gamma^\nu\right\} = \left[\Delta(\sig), \gamma^5\right] = 0$ to compute
 \begin{align}
  \label{eq:15}
 g_{\mu\nu}\alpha^\mu_{\sig}\alpha^{\nu}_{\sig}= -
  g_{\mu\nu}\gamma^\mu\gamma^5\Delta(\sig)
  \gamma^\nu\gamma^5\Delta(\sig)=  -
  g_{\mu\nu}\gamma^\mu\gamma^\nu\Delta^2(\sig)   = -4\Delta^2(\sig).
\end{align}
Thus 
$\beta_\sig - \frac 14\alpha_{\sig}\cdot \alpha_{\sig} = \Delta^2(\sig)- \sig^2 D_R^2
$
and \eqref{eq:178-2} follows from \eqref{eq:73-1} and \eqref{eq:48}.
\end{proof}
\medskip

In order to interpret proposition \ref{Lichnedetail}, it is instructive
  to confront with the
  non-twisted case.  When the finite dimensional algebra $\A_F$ of an almost-commutative geometry acts trivially on
spinors, the full covariant Dirac operator is
\begin{equation}
  \label{eq:185}
  D_A = D_Y + \gamma^5\otimes D_R,
\end{equation}
where
$D_Y := -i\gamma^\mu \nabla^Y_\mu\label{eq:163}
$
is the covariant Dirac operator of a $U(\A_F)$-bundle over the spin
bundle of $\M$, associated with the covariant derivative
$\nabla_\mu^Y:= \nabla_\mu^S + {\mathbb Y}_\mu$
defined by a connection one-form $\mathbb{Y}_\mu
= \delta_{s\dot s}^{t\dot t}\, Y_{\sJ\alpha\sC}^{I\beta\sD}$ whose action on spinor indices is trivial. One gets
\begin{equation}
  \label{eq:191}
   D_A^2 = D_Y ^ 2 + D_R^2 + \left\{\slashed D_Y, \gamma^5\otimes D_R\right\}.
\end{equation}
Because 
\begin{equation}
[\gamma^\mu, \mathbb{Y}_\nu]= 0,
\label{eq:162}
\end{equation}
one has
    $D_Y^2 = - \gamma^\mu \gamma^\nu \nabla^Y_\mu \nabla^Y_\nu$.
Using 
$\gamma^\mu\gamma^\nu = g^{\mu\nu}\mathbb{I}  + \frac 12
\gamma^\mu\gamma^\nu - \frac 12 \gamma^\nu\gamma^\mu$, the square of
$D_Y$ is rewritten
as the sum of the Laplacian and the field strength 
$F_{\mu\nu} =\left[{\mathbb Y}_\mu, \mathbb{Y}_\nu\right]$, namely 
\begin{align}
  \label{eq:158}
 D_Y^2 = 
  -g^{\mu\nu} \nabla_\mu^Y \nabla_\nu^Y - \frac 12\gamma^\mu \gamma^\nu F_{\mu\nu}.
\end{align}
The second term in \eqref{eq:191} is $D_R^2=
|k_R|^2\mathbb{I}$ and the third is  $-i\gamma^\mu\gamma^5
[\mathbb{Y}_\mu,D_R]$. Therefore the Lichnerowicz formula for the
covariant \emph{non-twisted} Dirac operator is
\begin{equation}
  \label{eq:186}
  D_A^2= - g^{\mu\nu}\nabla_\mu^Y\nabla_\nu^Y - \frac 12\gamma^\mu\gamma^\nu
    F_{\mu\nu} + |k_R|^2{\bb I} - i\gamma^\mu\gamma^5 [{\mathbb Y}_\mu,D_R].
\end{equation}

In the twisted case, summing up the terms in Prop. \ref{Lichnedetail} one
obtains from \eqref{eq:173}
\begin{align}
  \label{eq:160}
  D_{\mathbb A}^2 &= -g^{\mu\nu}\nabla_\mu\nabla_\nu 
-\frac 12
    \gamma^\mu\gamma^\nu\left(\mathbb{F}_{\mu\nu} + D_\nu\, \Delta(\mathbb
    X)_\mu + \Delta(\XX)_\mu\Delta(\XX_\nu)\right) \\
  \label{eq:160-1}
& + \sig^2D_R^2 - \Delta(\sig)^2 -
  i\gamma^\mu\gamma^5\left( D_\mu (\sig D_R) - \frac 12 D_\mu
    \Delta(\sig)\right) \\
  \label{eq:160-2}
& - \frac i2
  \gamma^\mu\gamma^5\left\{3\XX_\mu - \Delta(\XX)_\mu, \Delta(\sig)\right\}.  
\end{align}
There are several important differences with the non-twisted case:
\begin{itemize}
\item In \eqref{eq:160}, the covariant derivative of $\Delta(\XX)_\nu$ and its potential
$\Delta(\XX)_\mu\Delta(\XX)_\nu$ can be traced back  to the Lichnerowicz formula for the twisted-covariant
free Dirac operator, 
\begin{equation}
  D_X^2= -g^{\mu\nu} \nabla^X_\mu \nabla^X_\nu - E_X\, \text{ where } \quad 
\nabla^X_\mu:= \nabla_\mu^S + \frac 12 \alpha_\mu^X.
\label{eq:172}
\end{equation}
These new terms arise because in the twisted case \eqref{eq:162} no
longer holds, instead one has
\eqref{eq:76}, that is
\begin{equation}
  \label{eq:134}
  [\gamma^\mu, X_\nu]_\rho = [X_\nu, \gamma^\mu]_\rho = 0.
\end{equation}

\item The appearance of the covariant derivative of $\sig D_R$ in
  \eqref{eq:160-1} is not surprising. It is already there in \eqref{eq:186}, where the last
  term  is nothing but the covariant derivative of $\sig D_R$ for $\sig$ the constant field $1$. Similarly $|k_R|^2$ in \eqref{eq:186} is the potential term $\sig^2D_R^2$  in \eqref{eq:160-1}  for $\sig=1$.

\item  In \eqref{eq:160-1} the scalar $\Delta(\sig)$ is described by  a dynamical term
  $-D_\mu \Delta(\sig)$ and a potential $-\Delta(\sig)^2$, whose sign
  are opposite to the similar terms for $\sig$.

\item  The interaction between $\XX_\mu$ and $\Delta(\sig)$ is not
  totally absorbed in the covariant derivative $D_\mu$. There remains
  in  \eqref{eq:160-2}
  an potential of interaction $\left\{3\XX_\mu, \Delta(\sig)
  \right\}$. As well, there is a potential of interaction $\left\{\Delta(\XX)_\mu,
    \Delta(\sig) \right\}$ between the extra scalar field and the
  additional vector field.
\end{itemize}

One may be puzzled by the presence of two distinct covariant
derivatives in the Lichnerowicz formula for $D_{\bb A}$:  $\nabla_\mu$
in the Laplacian and $D_\mu$ that encodes the dynamics of the fields $\Delta(\XX)_\mu$ and $\Delta(\sig)$.
In the non-twisted
case this is the same covariant derivative 
$\nabla_\mu^Y$ which plays both roles. However, because we switch 
gravitation off{\footnote{Our aim in this paper is to understand how the twist allows to generate the
field $\sig$. That is why for simplicity  we consider the flat
case. The curved case, which should be similar,  will be studied elsewhere.}} and consider
the flat case, in the heat kernel expansion of the spectral action
the covariant derivative $\nabla_\mu$ only appears through the term
$\nabla^\mu\nabla_\mu E$ (in $a_4$). The latter is interpreted as a
boundary term (see \cite[Rem.1.155]{Connes:2008kx}) and we shall not take it into account in this
paper. Doing so, only one covariant derivative remains, $D_\mu$.
This makes sense from our perspective:  the
fields $\Delta(\XX)_\mu$ and $\Delta(\sig)$ are viewed as ``excitations''
generated by the twist,  living on a background gauge
theory with connection $1$-form $X_\mu$;  so their dynamics is encoded
by $D_\mu$, not by $\nabla_\mu$. 

The remaining Seeley-de Witt
coefficients are $a_0$, which is not affected by the twist and is interpreted as the
cosmological constant (which recently turns out to be quantized, see \cite{quanta-geom})
and the integral of the trace of  $E$ (in $a_2$) and $E^2$ (in $a_4$)
for $E$ given in \eqref{eq:70}. In other terms the potential is the part of 
\begin{equation}
  \label{eq:175}
  V:= \Lambda^2 f_2 \;\text{Tr}\, E \;+ \;\frac 12\, f_0 \text{Tr}\, E^2
\end{equation}
that does not depend on the covariant derivative $D_\mu$. We analyze
this potential below, dividing it into three pieces: the potential $V(X)$ of the vector
field, $V(\sig)$ of the scalar field, and a potential of interaction $V(X,\sig)$.  

\subsection{The vector field and the breaking to the standard
model}
\label{secbreaking}

The potential
$V(X)$ is the part of $V$ that depends on $\Delta(\XX)_\mu$
and not on its derivative, that is
\begin{equation}
  \label{eq:184}
  V(X) = \Lambda^2 f_2 \,\text{Tr}\, E_X^0 + \frac 12\, f_0
  \text{Tr}\,  (E_X^0)^2,
\end{equation}
where $  E_X^0 := \frac 12 \gamma^\mu\gamma^\nu \Delta(\XX)_\mu
\Delta(\XX)_\nu$ is read in \eqref{eq:178}. One rewrites it as
\begin{equation}
  \label{eq:183}
  E_X^0 = \frac 12 \slashed\Delta^2(\XX),
\end{equation}
thanks to \eqref{eq:146} which guarantees that $\gamma^\nu$ anti-commutes with $\Delta(\XX)_\mu$
for all $\mu$.
\begin{prop}
  The potential $V(X)$ is never negative and vanishes iff
  $\Delta(\XX)_\mu =0$ for any~$\mu$.
\end{prop}
\begin{proof}
Since $\slashed\Delta(\XX) $ is selfadjoint, $E_X^0$ and $(E_X^0)^2$ are
positive. Thus their trace is never negative, and vanishes if and only
if $E_X^0=(E_X^0)^2=0$. This condition is equivalent to 
\begin{equation}
  \label{eq:181}
  \Delta(\XX)_\mu = 0 \quad \forall \mu. 
\end{equation}
Indeed, since $\left\{\gamma^\nu,\Delta(\XX)_\mu\right\}=0$ one has
\begin{equation}
  \label{eq:175bis}
  \text{Tr }(\gamma^\mu\gamma^\nu \Delta(\XX)_\mu \Delta(\XX)_\nu) =
  \text{Tr }(\gamma^\mu \Delta(\XX)_\mu \Delta(\XX)_\nu\gamma^\nu) =   \text{Tr }(\gamma^\nu\gamma^\mu \Delta(\XX)_\mu \Delta(\XX)_\nu)
\end{equation}
where the last equality comes from the tracial property. Therefore
\begin{align}
  \label{eq:188}
  \text{Tr} \, E_X^0 &=  \frac 14 (\text{Tr }(\gamma^\mu\gamma^\nu
  \Delta(\XX)_\mu \Delta(\XX)_\nu) +  \text{Tr }(\gamma^\nu\gamma^\mu
  \Delta(\XX)_\mu \Delta(\XX)_\nu)),\\
\label{eq:188-bis}
&= \frac 12 \, g^{\mu\nu} \,\text{Tr}(\Delta(\XX)_\mu \Delta(\XX)_\nu)
=\frac 12\sum_\mu  \text{Tr} \left(\Delta^2(\XX)_\mu\right) .
\end{align}
Since $\Delta(\XX)_\mu$ is selfadjoint, $\Delta^2(\XX)_\mu$ is
positive. Its trace is never negative and vanishes if and only if
$\Delta(\XX)_\mu$ is zero. The same is true for the sum in
\eqref{eq:188-bis}, meaning that $\text{Tr } E_X^0$ - hence $E_X^0$~- vanishes if and only if
$\Delta(\XX)_\mu = 0$  for all  $\mu$. 

The
proposition is obtained noticing that  $f_0$ and $f_2$ are positive numbers.
\end{proof}

Condition \eqref{eq:181} is equivalent to
$\Delta(X)_\mu=0$ for any $\mu$. To obtain the breaking to the standard model, one needs to check
that the vanishing of $\Delta(X)_\mu$, that is
the invariance of $X_\mu$ under the twist, implies the invariance of its
components $R^i, Q_i$.
\begin{lemma}
\label{propinv}
  The biggest unital subalgebra of $\cal B\otimes \cinf$ for which any
  combination
  \begin{equation}
X_\mu=\delta^I_J\,\rho(R^i)\,\del_\mu Q_i - \delta_\alpha^\beta
  \,\bar N^i \del_\mu \bar  M_i
\label{eq:21}
  \end{equation}
is invariant under the twist is
  $\A_{SM}\otimes C^\infty(\M)$.
\end{lemma}
\begin{proof} Let $\cal  G$ be any subalgebra of $\cal B\otimes \cinf$
  such that any  linear
  combination $X_\mu$ with $(R^i, N^i)$ and  $(Q_i, M_i)$ in $\cal
 G$ is invariant under the automorphism $\rho$. This means in
 particular that for $X =
 R\del_\mu Q - Q\del_\mu R$ with $R, Q$ arbitrary elements in $\cal
 G$, one has 
  \begin{equation}
    \label{eq:20}
    \rho(X_\mu) - X_\mu = \rho(R) \del_\mu Q - R
    \del_\mu\rho(Q) = 0.
  \end{equation}
Taking $R=\mathbb I$, this implies 
\begin{equation}
  \label{eq:19}
  \del_\mu (Q - \rho(Q)) = 0.
\end{equation}
So any element of $\cal G$ is $(Q, M)$ where
\begin{equation}
  \label{eq:7}
  Q =\left(\begin{array}{cc}
Q_r^r & 0 \\ 0 &Q_r^r + c
\end{array}\right)_{st}
\end{equation}
with  $c$ a constant.
For $\cal G$ to be an algebra, \eqref{eq:7} must be true also for
$Q^2$, that is there must exists a constant $c'$ such that 
\begin{equation}
  \label{eq:7bis}
  Q^2 =\left(\begin{array}{cc}
{(Q_r^r)}^2 & 0 \\ 0 & ({Q_r^r})^2 + c^2 + 2cQ_r^r 
\end{array}\right)_{st} =\left(\begin{array}{cc}
{(Q_r^r)}^2 & 0 \\ 0 & ({Q_r^r})^2 + {c'}^2 
\end{array}\right)_{st}. 
\end{equation}
This is possible if and only if $c=c'=0$. Thus $\rho(Q) = Q$ for any
$(Q, M)\in \cal G$. The proposition follows from the identification of
$\A_{sm}$ as the biggest $\rho$-invariant sub-algebra of $\cal B$.
\end{proof}
\noindent This proves the third statement of theorem \ref{theo1}, namely the
breaking of grand symmetry to the standard model is
  dynamical, and induced by the minimal of the spectral action of the
  twisted-covariant free Dirac operator $D_X$.  

 \subsection{The scalar field}
 \label{section:potentialsigma}

The part of the potential containing only the extra scalar field
and not the vector field  is 
\begin{equation}
  \label{eq:170}
    V(\sig):= \Lambda^2 f_2 \,\text{Tr}\, E^0_\sig +  \frac 12 f_0 \,\text{Tr}\, (E_\sig^0)^2,
\end{equation}
where
\begin{equation}
E_\sig^0 := \Delta^2(\sig)-\sigma^2D_R^2
\label{eq:65}
\end{equation}
is read in
\eqref{eq:178-2}. Compared to $V(X)$ which contains only
$\Delta(\XX)_\mu$ and not $\XX_\mu$, the potential $V(\sig)$ contains both $\sig$
and $\Delta(\sig)$. This gives two possibilites for minimizing:

Either one considers only $\Delta(\sig)$ as degree of freedom. The
potential then reduces to
\begin{equation}
  \label{eq:197}
  V(\Delta(\sig)):= \Lambda^2\, f_2\,\text{Tr}(\Delta^2(\sig)) + \frac 12 f_0\,\text{Tr}(\Delta^4(\sig)).
\end{equation}
Since $\Delta(\sig)$ is selfadjoint, this potential is positive and
vanishes if and only if $\Delta(\sig) = 0$. Going back to the the definition \eqref{eq:171} 
of $\Delta(\sig)$, this means
 \begin{equation}
\sig = \rho({\sig}) =
  \mathbb I.\label{eq:131}
  \end{equation}

Or one may prefer to take into account the whole potential
\eqref{eq:170}, including the term in $\sig$. In this case it is easier to take as degree of freedom the field
$\phi$.
\begin{lemma}
\label{propscal}
The potential of the scalar field is
  \begin{equation}
  \label{eq:18}
  V(\sig) = C_4 \,\phi^4   + C_2\,\phi^2 + C_0
  \end{equation}
where
$C_4:= 36 |k_R|^4 f_0 ,\; C_2:=8 |k_R|^2 (3\Lambda^2 f_2  - |k_R|^2
  f_0 ),\; C_0:= 8 |k_R|^2 \left( \frac{|k_R|^2}{2}f_0  - \Lambda^2 f_2\right)$.
\end{lemma}
\begin{proof}
By \eqref{eq:17}, \eqref{eq:66} and  \eqref{eq:171}  one has
\begin{equation}
\label{eq:135}
  E_\sig^0 = \left((3\phi^2 - 1){\mathbb I}_4 -  2\gamma^5\phi\right)\, D_R^2.
\end{equation}  From \eqref{majoranadiracoperator},
\begin{equation}
  \label{eq:83}
  D_R^2 = |k_R|^2 \, \delta_{\sC}^{\sD}\, \Xi_{\sJ\alpha}^{\sI\beta}
\end{equation}
so that $\gamma^5D_R^2$ has zero trace. Hence
\begin{align}
  \text{Tr} \, E_{\sig}^0 &= (3\phi^2 - 1) \;\text{Tr}\left( \;
    {\mathbb I}_4 \otimes D_R^2\right) = 8|k_R|^2\,(3\phi^2 -1). 
\end{align}
Squaring \eqref{eq:135} one gets
\begin{align}
  \label{eq:55}
  (E_{\sig}^0)^2 &= \left((3\phi^2 - 1)^2{\mathbb I}_4 +  4\phi^2 {\mathbb I}_4 - 2\gamma^5\phi(3\phi^2-1)\right)\, D_R^4
\end{align}
whose trace is
\begin{align}
  \label{eq:57}
  \text{Tr}\,   (E^0_{\sig})^2 = 8|k_R|^4\left((3\phi^2 - 1)^2+  4\phi^2 \right).
\end{align}
The result follows from \eqref{eq:170}.
\end{proof}
\noindent At a large unification scale $\Lambda$ it is reasonable to assume
that (see e.g. \cite[\S 11.3.2]{Walterlivre})
\begin{equation}
3\Lambda^2 f_2 \geq f_0|k_R|^2.
\label{eq:168}
\end{equation}
 Together with the positivity
of $f_0$ this shows that 
  $V(\sig)$ is minimum when $\phi=0$, and one is led to the same
  conclusion \eqref{eq:131} obtained by minimizing $V(\Delta(\sig))$.

The invariance \eqref{eq:131} of $\sig$ under the twist implies that
$D_\sig =D_M$, so that one is back to the Dirac operator 
of  the standard model. However this does not imply the
  reduction of the algebra to the one of the standard model. Indeed,
from \eqref{eq:13} the vanishing of $\phi$ means $c_R^r =c_R^l$,
  so that the bigger subalgebra of
  $\cinf\otimes \cal B$ for which any fluctuation yields a
  $\rho$-invariant $\sig$ is 
  \begin{equation}
    \label{eq:52}
    \cinf\otimes ({\mathbb H}_L^l \oplus {\mathbb H}_L^r \oplus {\mathbb C}
    \oplus M_3({\mathbb C})),
  \end{equation}
which contains, but is different from
  $\cinf\otimes \A_{sm}$.

\subsection{Potential of interaction}
\label{interaction}

We now consider the remaining part of the potential, that is the
interaction term $V(X,\sig)$ between the scalar and the vector fields. Writing
$V(X,\sig)$ explicitly in \eqref{eq:98} below, it will become clear it
is easier to minimize it together with the potential $V(X)$,
which will be done in prop. \ref{propinteract}.

Let 
\begin{equation}
  \label{eq:24}
  E_{X\sig}^0:=  \frac i2 \gamma^\mu\gamma^5\left\{ {\mathbb H}_\mu,
    \Delta(\sig)\right\}\quad\text{ with } \quad {\mathbb H}_\mu:=3\XX_\mu -\Delta(\XX)_\mu
\end{equation}
denote the part  of $E_{X\sig}$ in \eqref{eq:178-1} that does not depend on the covariant
derivative of the fields. The potential of interaction is made of all the terms in the trace of
$E^0_X+ E^0_\sig + E^0_{X\sig}$ and its square that depend on both $X$ and
$\sig$.  

\begin{lemma}
\label{lemmexsig} $E^0_{X\sig}$ is selfadjoint and traceless.
\end{lemma}
\begin{proof}
Since  $\gamma^\mu$ anti-commutes with $\Delta(\sig)$ and $\gamma^5$, one has 
\begin{equation}
  \label{eq:90}
  E_{X\sig}^0= \frac 12 \gamma^5  \left[\slashed{\mathbb H}, \Delta(\sig)\right].
\end{equation}
 Therefore 
\begin{equation}
\text{Tr}\, E_{X\sig}^0 = \frac 12 \text{Tr}\, \gamma^5
\left[\slashed{\mathbb H}, \Delta(\sig)\right] = -\frac 12 \text{Tr} 
\left[\slashed{\mathbb H}, \Delta(\sig)\right]\gamma^5 =  -\frac 12 \text{Tr} \,\gamma^5
\left[\slashed{\mathbb H}, \Delta(\sig)\right]
\label{eq:32}
    \end{equation}
  where the first equality come from $\left\{\gamma^5, \slashed {\mathbb H}\right\}= \left[\gamma^5, \Delta(\sig)\right] = 0$ and the second from the tracial
      property. Thus $\text{Tr}\, E_{X\sig}^0  = - \text{Tr}\,
      E_{X\sig}^0,$ and so vanishes. The selfadjointness follows from
      the commutation properties of $\gamma^5$ and the selfadjointness
      of $\slashed{\mathbb H}$ and $\Delta(\sig)$.
\end{proof}
\noindent Furthermore $E_X$ and $E_\sig$ depend solely on $X$ and $\sig$, so the
potential of interaction
reduces to the trace of the part of $(E^0_X + E^0_\sig + E^0_{X\sig})^2$
that contains products of $X$ and $\sig$, namely
\begin{equation}
  \label{eq:44}
E_{X\sig}^2 + \left\{E^0_X, E^0_\sig\right\} +\left\{E^0_X,
  E^0_{X\sig}\right\} +  \left\{E^0_\sig, E^0_{X\sig}\right\}.
\end{equation} 
The last two terms are traceless because $E^0_X$ and $E^0_\sig$ are diagonal in
the $\sC\sD$ indices (see (\ref{eq:183}, \ref{eq:1999}, \ref{eq:199})
and (\ref{eq:135}, \ref{eq:83})) while $E^0_{X\sig}$ is off-diagonal
(being the commutator of a diagonal and an off-diagonal matrice). By the tracial property
$\text{Tr}\, \left\{E^0_X, E^0_\sig\right\}  = 2 \,\text{Tr}\,E^0_X
E^0_\sig$, hence the potential of interaction is 
\begin{equation}
  \label{eq:98}
  V(X,\sig) := \frac 12 f_0 \,\text{Tr}\,(E_{X\sig}^0)^2 +  f_0
    \,\text{Tr}\, E_X^0 E_\sig^0. 
\end{equation}

\begin{lemma}
\label{lemmainteract}
One has  \begin{equation}
    \label{eq:77}
    \text{Tr} \, E_X^0 E_\sig^0 = \frac 12 (3\phi^2 -1) \text{Tr}\left(\slashed\Delta^2(\XX) D_R^2\right).
  \end{equation}
\end{lemma}
\begin{proof}
By \eqref{eq:183} and \eqref{eq:135} one has
\begin{equation}
  \label{eq:68}
  E_X^0 E_\sig^0 = \frac 12 (3\phi^2 -1) \slashed\Delta^2(\XX) D_R^2 -
  \phi \,\slashed\Delta^2(\XX) \gamma^5 D_R^2.
\end{equation}
The result amounts to show that the second term has vanishing trace. To see it, let use \eqref{eq:166},
\eqref{eq:1999} and \eqref{eq:83} to write
\begin{equation}
  \label{eq:202}
  \slashed\Delta^2(\XX) \gamma^5 D_R^2=|k_R|^2\left(
\begin{array}{cc}
  \slashed\Delta^2(X)\,\gamma^5\, \Xi_{\alpha\sI}^{\beta\sJ}& 0_{64}\\ 
  0_{64}& -\slashed{\Delta}^2(\bar X)\,\gamma^5
  \,\Xi_{\alpha\sI}^{\beta\sJ}
\end{array}\right)_{\sC\sD}.
\end{equation}
One has
\begin{equation}
  \label{eq:204}
   {\Delta(X)}_\mu = \delta_{\dot s}^{\dot t}\left(\begin{array}{cc}
 X_\mu^l -  X_\mu^r & 0_{32}\\
0_{32} & X_\mu^r - X_\mu^l\end{array}\right)_{st} =:
\delta_{\dot s}^{\dot t}\left(\begin{array}{cc}
\Delta_\mu & 0_{32}\\
0_{32} & -\Delta_\mu\end{array}\right)_{st},
\end{equation}
so that
\begin{equation}
 \slashed\Delta(X) = -i\left(\begin{array}{cc}
     0_{32}& -\sigma^\mu \Delta_\mu\\
     \bar\sigma^\mu \Delta_\mu &0_{32} 
\end{array}\right)_{st},\quad
  {\slashed \Delta^2(X)} = \left(\begin{array}{cc}
      \sigma^\mu\tilde\sigma^\nu\Delta_\nu \Delta_\mu & 0_{32}\\
      0_{32} & \tilde \sigma^\mu \sigma^\nu \Delta_\mu\Delta_\nu
\end{array}\right)_{st}.
\end{equation}
Hence
\begin{align}
  \label{eq:206}
  \text{Tr}\left(\slashed\Delta^2(X)\,\gamma^5\,\Xi_{\alpha\sI}^{\beta\sJ}
  \right)&= 
\text{Tr}\left(
  \begin{array}{cc}
    \sigma^\mu\tilde\sigma^\nu X_\mu^l X_\nu^r & 0_{32}\\
    0_{32} & -\tilde \sigma^\mu \sigma^\nu X_\mu^r X_\nu^l
\end{array}\right)_{st}\\
\nonumber
&= \text{Tr}\left(\sigma^\mu\tilde\sigma^\nu X_\mu^l X_\nu^r \right)
   - \text{Tr}\left(\tilde \sigma^\mu \sigma^\nu X_\mu^r X_\nu^l\right) 
= \text{Tr}\left(\sigma^\mu\tilde\sigma^\nu X_\mu^l  X_\nu^r \right)
- \text{Tr}\left(\tilde \sigma^\nu \sigma^\mu X_\nu^r X_\mu^l\right)
\end{align}
which vanishes by the trace property and the commutation of
$[X_\mu^l,\sigma^\mu] = [X_\mu^l,\sigma^\nu]=0$.  The same is true for
$\text{Tr}\left(\slashed\Delta^2(\bar X)\,\gamma^5\,\Xi_{\alpha\sI}^{\beta\sJ}\right)$, so that \eqref{eq:202}
  has zero trace.
  \end{proof}

By lemma \ref{lemmexsig}, $(E_{X\sig}^0)^2$ is positive, hence its
trace is never negative and minimal when $E_{X\sig}^0$ is
zero. However $\text{Tr}\, E_X^0 E_{X\sig}^0$ is not necessarily
bounded from below, which makes difficult to minimize the potential of
interaction alone. In fact it is easier to minimize it together with the
potential $V(X)$ of the vector field.
\begin{prop}\label{propinteract}
  The potential
  \begin{equation}
    \label{eq:152}
      V'(X,\sig):=  V(X) + V(X,\sig)
\end{equation}
is never negative and vanishes if and only if $\Delta(\XX)_\mu=0$ for
all $\mu$.\end{prop}
\begin{proof}
We write $\slashed\Delta$ for $\slashed\Delta(\XX)$. By summing up \eqref{eq:184} and %, \eqref{eq:18} and 
\eqref{eq:98} one gets
\begin{align}
  \label{eq:159}
  V'(X,\sig) &= \frac{\Lambda^2 f_2}{2} \text{Tr}\,\slashed\Delta^2 +
  \frac {f_0}8 \text{Tr}\,\slashed\Delta^4 +\frac {f_0}2\text{Tr}\,(E_{X\sig})^2 + \frac{f_0}2
  (3\phi^2-1)\text{Tr}\,(\slashed\Delta^2D_R^2),\\
  \label{eq:159-bis}
&=\text{Tr}\left( \frac{\Lambda^2 f_2}{2} \,\slashed\Delta^2 - \frac {f_0}2 \,
 \slashed\Delta^2D_R^2 \right)+  \frac {f_0}{8} \text{Tr}\,\slashed\Delta^4
  + \frac{f_0}2\text{Tr}\,(E_{X\sig})^2 + \frac{3f_0}4
  \phi^2\text{Tr}\,(\slashed\Delta^2D_R^2).%   \\
\end{align}

Let  $p:=
\delta_{st\sC}^{\dot t\dot s\sD}\,\Xi_{\sJ\alpha}^{\sI\beta}$ denote the
projection on the non-zero entries of $D_R^2$, so that $\delta_{st}^{\dot s\dot t}D_R^2 = |k_R|^2 p$. The first term in \eqref{eq:159-bis}  is
\begin{equation}
  \label{eq:164}
W(X,\sig):= \left( \frac{\Lambda^2 f_2}2 - \frac 12 f_0 |k_R|^2\right)\text{Tr}\,(\slashed\Delta^2\, p)
  +\frac{\Lambda^2 f_2}2 \text{Tr}\,(\slashed\Delta^2\,(1-p)).
\end{equation}
Because $\slashed\Delta$ is selfadjoint, 
\begin{equation}
  \label{eq:165}
  \text{Tr}\, (\slashed\Delta^2 p) = \text{Tr}(p \slashed\Delta^2 p) =
  \text{Tr}\left( (p\slashed\Delta) (p\slashed\Delta)^\dagger\right)
\end{equation}
is positive. The same is true for $\text{Tr}\,(\slashed\Delta^2(1-p))$. Assuming as in \eqref{eq:168} that at high energy
\begin{equation}
  \label{eq:63}
  \Lambda^2 f_2 \geq f_0|k_R|^2,
\end{equation}
one gets that $W(X,\sig)$ is never negative, and vanishes if and only
if 
$\slashed\Delta^2 p$ and $\slashed\Delta^2 (1-p)=0$, that is if and
only if
$\slashed\Delta^2=0$, which is equivalent to $\text{Tr} \,\Delta^2=0$
since
$\Delta^2$ is positive. By \eqref{eq:188-bis} this is equivalent to
$\Delta(\XX)_\mu=0$ for any
$\mu$.

The second term in \eqref{eq:159-bis} is never negative, and vanishes
when $\Delta(X)_\mu=0$ for any $\mu$.  The same is true for the third
term by lemma \ref{lemmainteract}, and for the last term since
\begin{equation}
  \label{eq:69}
  \text{Tr}\,(\slashed\Delta^2D_R^2)= \text{Tr}\left((D_R\slashed\Delta)(D_R\slashed\Delta)^\dagger\right).
\end{equation}
Hence $V'(X,\sig)$ is never negative, and vanishes if and only if
$\Delta(\XX)_\mu=0$ for any $\mu$.
\end{proof}
\medskip

Combining propositions \ref{propinteract} and \ref{propscal}, the whole potential $V(X) + V(\sig) + V(X, \sig)$ is zero if and
only if both the scalar field $\sig$ and the vector field
$\Delta(X)_\mu$ are zero. This proves the first statement of point iii) in theorem
\ref{theo1}. The second statement has been proven below lemma \ref{propscal}.
 
\section{Twist and representations}\setcounter{equation}{0}
\label{section:discussion}

We discuss the choices made in the construction of the twisted
spectral triple of the standard model: the middle-term solution consisting in imposing by hand
the reduction $M_8(\bb C)\to~M_4(\bb C)$, and the representation
of $\A_G$.

\subsection{Global twist}
\label{secdisc}

Instead of reducing by hand ${\cal B}_{LR}$ to $\cal B'$ by imposing the reduction
$M_8(\bb C) \to M_4(\bb C)$, one could twist ${\cal B}_{LR}$ as well. This means finding an automorphism $\rho$
of $M_8(\bb C)$ such that
\begin{equation}
\label{discussion1}
\sigma^{\mu} M\,\partial_{\mu}-\rho(M)\sigma^{\mu}\partial_{\mu} = 0,\quad\quad
\bar{\sigma}^{\mu} M\,\partial_{\mu}-\bar{\rho}(M)\bar{\sigma}^{\mu}\partial_{\mu}=0.
\end{equation}
Using 
$\sigma^{\mu}\bar{\sigma}^{\nu}\partial_{\mu}\partial_{\nu}=\nabla^{2}$,
the first expression yields
\begin{equation}
\rho(M)=\sigma^{\mu}M\bar{\sigma}^{\nu}\frac{1}{\nabla^{2}}\partial_{\mu}\partial_{\nu}.
\label{sigmaM def}
\end{equation}

This does not define an automorphism of
$\cinf\otimes \A_G$. Indeed, 
writing $T_{\mu\nu}\equiv\frac{1}{\nabla^{2}}\partial_{\mu}\partial_{\nu}$
and $M_{1}^{\mu\nu}\equiv\sigma^{\mu}M_{1}\bar{\sigma}^{\nu},$
one gets
\begin{eqnarray}
\rho(M_{1})\rho(M_{2}) & = & \left(M_{1}^{\mu\nu}T_{\mu\nu}\right)\left(M_{2}^{\alpha\beta}T_{\alpha\beta}\right)\\
 & = &
 M_{1}^{\mu\nu}\left[T_{\mu\nu},M_{2}^{\alpha\beta}\right]T_{\alpha\beta}+M_{1}^{\mu\nu}M_{2}^{\alpha\beta}T_{\mu\nu}T_{\alpha\beta},\\
&=& \rho(M_{1}M_{2})+M_{1}^{\mu\nu}\left[T_{\mu\nu},M_{2}^{\alpha\beta}\right]T_{\alpha\beta}
\end{eqnarray}
where
we compute
\begin{eqnarray}
M_{1}^{\mu\nu}M_{2}^{\alpha\beta}T_{\mu\nu}T_{\alpha\beta} & = & \sigma^{\mu}M_{1}\bar{\sigma}^{\nu}\sigma^{\alpha}M_{2}\bar{\sigma}^{\beta}\frac{1}{\nabla^{2}}\frac{1}{\nabla^{2}}\partial_{\mu}\partial_{\nu}\partial_{\alpha}\partial_{\beta}\nonumber \\
 & = & \sigma^{\mu}M_{1}M_{2}\bar{\sigma}^{\beta}\frac{1}{\nabla^{2}}\partial_{\mu}\partial_{\beta}\nonumber \\
 & = & \rho(M_{1}M_{2}).
\end{eqnarray} 

A possible solution is to look for a
$\star$ product such that
\begin{equation}
\rho(M_{1})\star\rho(M_{2})=\rho(M_{1}\star M_{2}),\label{eq: star product}
\end{equation}
that would encode the intrinsic mixing between
the manifold (space-time) and the matrix part (gauge sector) that
is the core of the Grand Symmetry. This would also force us to
consider an algebra $\A_0$ of pseudo-differential operators bigger than
$\cinf\otimes \A_G$. This point is particularly interesting if one
believes that almost commutative geometries are an
effective low energy description of a more fundamental theory, based
on a ``truly'' non-commutative algebra (that is with a finite
dimensional center). This idea
has been often advertised by D. Kastler, and it could be that $\A_0$
is not so far from the ``noncommutative salmon'' he aimed at
fishing. All this will be investigated in future works.
\medskip

The reason why we choose the representation \eqref{eq:49}  instead of
\eqref{eq:49bis} as in \cite{Devastato:2013fk} is that while it is
right that \eqref{sigmaM def}
is still in $M_4(\bb C)$, it would
not be true for an element $Q=Q_{\dot s\alpha}^{\dot t\beta}\in
M_2(\bb H)$
that $\sigma^{\mu}Q\bar\sigma^\nu$ is still in $M_2(\bb H)$.
However, all the results presented in this paper would also be true
with the representation  \eqref{eq:49bis}, as explained in the
next paragraph.

\subsection{Invariance of the constraints}
\label{subsec:inv}

The grand algebra in the representation
\eqref{eq:49} is broken by  the grading to \cite[eq. (3.17)]{Devastato:2013fk} 
\begin{equation}
  \label{eq:39}
  \A'_G = M_2(\bb H)_L \oplus M_2(\bb H)_R \oplus M_4^l(\bb C) \oplus M_4^r(\bb C).   
\end{equation}
To have bounded commutators with $\Ds$, we impose by hand that
quaternions act trivially on the $\dot s$ index, yielding the reduction to
\begin{equation}
  \label{eq:45}
\A':= \bb H_L  \oplus \bb H_R  \oplus M_4^l(\bb C) \oplus M_4^r(\bb C)
\end{equation}
whose elements are $(Q, M)$ where 
\begin{equation}
  \label{eq:59}
   Q= \delta_{s\dot s}^{t\dot t} 
   \left(\begin{array}{cc} 
       q_R& 0_2\\
       0_2 &  q_L
     \end{array}\right)_{\alpha\beta}, \quad M = 
   \left(\begin{array}{cc}
       M_l^l & 0_4 \\ 
       0_4 & M_r^r
     \end{array}\right)_{st}
\quad\text{ with }
q_r\in{\mathbb H}, M_l^l, M_r^r \in M_4(\mathbb C).
\end{equation}
The twist $\rho$ is still defined as the exchange of the left and
right part of spinors, but it now acts on the matrix part
\begin{equation}
  \label{eq:61}
  \rho(M) = \left(
    \begin{array}{cc}
      M_r^r & 0_4 \\ 
      0_4 & M_l^l
    \end{array}\right)_{ s t}.
\end{equation}
This guarantees that 
\begin{equation}
  \label{eq:51}
[\Ds, M]_\rho = (\ds M) +[\gamma^\mu, M]_\rho   = (\ds M)
\end{equation}
is bounded, so that $(C^\infty(\M)\otimes \A', {\cal H}, \Ds + D_M; \rho)$ is a twisted
spectral triple. The twisted first-order condition for $\Ds$ is
checked as in proposition \ref{twisted-spec-triple}.

For the twisted first-order condition imposed by $D_M$, one first
consider the subalgebra of ${\cal A}'$ 
\begin{equation}
\tilde {\cal A}:=\bb H_L\oplus\bb C_R\oplus M_{3}^{l}(\bb C)\oplus\bb
C^{l}\oplus M_{3}^{r}(\bb C)\oplus \bb C^r
\label{eq:9601}
\end{equation}
obtained by  asking 
\begin{equation}
q_{R}=\left(\begin{array}{cc}
c_{R} & 0\\
0 & \bar{c}_{R}
\end{array}\right)\quad\text{ with } c_{R}\in\bb C
\end{equation}
in \eqref{eq:59} and
\begin{equation}
M_{r}^{r}=\left(\begin{array}{cc}
m^{r} & 0_{2}\\
0_{2} & {\bf M}^{r}
\end{array}\right)_{\sI\sJ},\quad M_{l}^{l}=\left(\begin{array}{cc}
m^{l} & 0_{2}\\
0_{2} & {\bf M}^{l}
\end{array}\right)_{\sI\sJ}\text{ with } {\bf M}^r, {\bf M}^l\in M_3(\bb C),\;
m^r, m^l\in\bb C.
\label{eq:100-1}
\end{equation}
Let $B=(R,N)\in\tilde{{\cal B}}$ be another element of $\tilde \A$, with components $d_r,
n^{r},n^{l}\in\bb C$ and ${\bf N}^r,{\bf N}^l\in M_{3}(\bb C)$.
The double twisted commutator  $[[D_M, A]_\rho , \; JBJ^{-1}]_\rho$
is an off-diagonal matrix with components
\begin{align}
  \label{eq:64}
&(\sD_M M - Q\sD_M)\bar R - \rho(\bar N) (\sD_M M -
  Q\sD_M),\\
  \label{eq:bis64}
& (\sD_M Q - \rho(M) \sD_M) \bar N
  - \bar R(\sD_M Q - \rho(M)\sD_M). 
\end{align}
One has
\begin{align}
\rho(\bar N ){\sf D_{\nu}}M & =(\rho(\bar N)\eta\Xi M)_{s\sJ}^{t\sI}(\Xi\delta)_{\dot{s}\alpha}^{\dot{t}\beta}=
\left(
  \begin{array}{cc}
    \bar{\sf n}^l\sf{m}^r & 0_{4}\\
    0_{4} & -\bar{\sf n}^r {\sf m}^l
  \end{array}
\right)_{st}\otimes
\left(
  \begin{array}{cc}
    \Xi & 0_{4}\\
    0_{4} &\Xi
  \end{array}
\right)_{\dot s\dot t},\label{eq:102}\\[5pt]
\rho(\bar{N})Q\sf D_{\nu} & =(\rho(\bar N)\eta\Xi)_{s\sJ}^{t\sI}\,(Q\Xi)_{\dot s\alpha}^{\dot t\beta}=
\left(
  \begin{array}{cc}
    \bar{\sf n}^l & 0_{4}\\
      0_{4} & -\bar{\sf n}^r
    \end{array}\right)_{st}\otimes
  \left(
    \begin{array}{cc}
      {\sf c}_{R} & 0_{4}\\
      0_{4} & {\sf c}_{R}
    \end{array}\right)_{\dot s\dot t},\\[5pt]
{\sf D_{\nu}}M\bar{R} & =(\eta\Xi M)_{s\sJ}^{t\sI}\,(\Xi\bar{R})_{\dot{s}\alpha}^{\dot{t}\beta}=
\left(
  \begin{array}{cc}
    {\sf m}^{r} & 0_{4}\\
    0_{4} & -{\sf m}^l
  \end{array}\right)_{st}\otimes
\left(
  \begin{array}{cc}
    \bar{\sf d}_{R} & 0_{4}\\
    0_{4} & \bar{\sf d}_{R}
  \end{array}
\right)_{\dot{s}\dot{t}},\\[5pt]
Q{\sf D}_{\nu}\bar{R} &
=(\eta\Xi)_{s\sJ}^{t\sI}\,(Q\Xi\bar{R})_{\dot{s}\alpha}^{\dot{t}\beta}=
\left(
  \begin{array}{cc}
    \Xi & 0_{4}\\
    0_{4} & -\Xi
  \end{array}\right)_{st}\otimes
\left(
  \begin{array}{cc}
    {\sf c}_{R}{\bar{\sf d}_{R}} & 0_{4}\\
    0_{4} & {\sf c}_{R}\bar{\sf d}_{R}
  \end{array}
\right)_{\dot{s}\dot{t}},
\end{align}
where we defined 
\begin{equation}
{\sf m^{r}}:=\left(\begin{array}{cc}
m^{r} & 0\\
0 & 0_{3}
\end{array}\right)_{\alpha\beta},\quad
{\sf m^{l}}:=\left(\begin{array}{cc}
m^{l} & 0\\
0 & 0_{3}
\end{array}\right)_{\alpha\beta},
\quad
{\sf c}_{R}=\left(\begin{array}{cc}
c_{R} & 0\\
0 & 0_{3}
\end{array}\right)_{\sI\sJ}
\label{eq:104bis}
\end{equation}
and similarly for ${\sf n}^r$, ${\sf n}^l$ and ${\sf d}_R$. Collecting the various terms,
one finds that \eqref{eq:64} is zero if and only if 
\begin{equation}
(c_R-m^r)(\bar{d}_R-\bar{n}^l)=0,\quad(c_R-m^l)(\bar{d}_R-\bar{n}^r)=0
\label{eq:105-bis}
\end{equation}
which are the same constraints \eqref{eq:105} coming from the other representation.
The same is true for \eqref{eq:bis64}, using
\begin{align}
\bar{R}\rho(M)\sf D_{\nu} & =(\rho(M)\eta\Xi)_{s\sJ}^{t\sI}\,(\Xi\bar{R})_{\dot{s}\alpha}^{\dot{t}\beta}=\left(\begin{array}{cc}
{\sf m}^l & 0_{4}\\
0_{4} & -{\sf m}^r
\end{array}\right)_{st}\otimes\left(\begin{array}{cc}
\bar{\sf d}_{R} & 0_{4}\\
0_{4} & \bar{\sf d}_{R}
\end{array}\right)_{\dot{s}\dot{t}},\label{eq:102bisbis}\\[5pt]
\bar{R}\sD_{\nu}Q & =(\eta\Xi)_{s\sJ}^{t\sI}\,(\bar R\Xi Q)_{\dot{s}\alpha}^{\dot{t}\beta}=\left(\begin{array}{cc}
\Xi & 0_{4}\\
0_{4} & -\Xi
\end{array}\right)_{st}\otimes\left(\begin{array}{cc}
{\sf c}_R\bar{\sf d}_R & 0_{4}\\
0_{4} & {\sf c}_R\bar{\sf d}_R
\end{array}\right)_{\dot{s}\dot{t}},\\[5pt]
\rho(M)\sD_{\nu}\bar{N} & =(\rho(M)\eta\Xi\bar{N})_{s\sJ}^{t\sI}(\Xi)_{\dot{s}\alpha}^{\dot{t}\beta}=\left(\begin{array}{cc}
 {\sf m}^l \,\bar{\sf n}^r& 0_{4}\\
0_{4} & -{\sf m}^r\,\bar{\sf n}^l
\end{array}\right)_{st}\otimes\left(\begin{array}{cc}
\Xi & 0_{4}\\
0_{4} & \Xi
\end{array}\right)_{\dot{s}\dot{t}},\\[5pt]
\sD_{\nu} Q\bar{N} & =(\eta\Xi \bar{N})_{s\sJ}^{t\sI}\,(\Xi Q)_{\dot{s}\alpha}^{\dot{t}\beta}=\left(\begin{array}{cc}
\bar{\sf n}^r & 0_{4}\\
0_{4} & -\bar{\sf n}^l
\end{array}\right)_{st}\otimes\left(\begin{array}{cc}
{\sf c}_{R} & 0_{4}\\
0_{4} & {\sf c}_{R}
\end{array}\right)_{\dot{s}\dot{t}}.
\end{align}
 Solving \eqref{eq:105} by
asking $m^r = c_R$, that is identifying $\bb C^{r}$ and  $\bb C_R$
with a single copy $\bb C_R^r$ of the complex numbers, one reduces
$\tilde A$ to 
\begin{equation}
  \label{eq:54}
  \A:= \mathbb H_L \oplus \bb C^r_R \oplus \bb C^l\oplus M_3^l(\bb C) \oplus
  M_3^r(\bb C). 
\end{equation}

This algebra plays for the representation \eqref{eq:49bis} the same role
as the algebra $\cal B$  for the representation
\eqref{eq:49}. Repeating the computation of \S\ref{soussectionsigma},
one finds a scalar field similar to~$\sig$. 
Thus, except for the hope of a global twist described in \S\ref{secdisc},
there is  at the moment no motivation to prefer one or the other of the two natural
representations of the grand algebra.
 
\section{Conclusion}\setcounter{equation}{0}

Let us summarize our results by the following chain of breaking, to be compared with \eqref{eq:50}:
\begin{framed}
  \begin{eqnarray*}
    \label{eq:final}
\mathcal{A}_{G}&=& M_4(\mathbb{H}) \oplus M_8(\bb C) \\[.25truecm] 
\nonumber 
&\Downarrow& \hspace{0truecm}\text{grading condition}\\[.25truecm] 
\nonumber
{\cal B}_{LR} &=&(\bb H^l_L\oplus \bb H^r_L\oplus \bb H^l_R\oplus \bb H^r_R) \oplus M_8(\bb C)\\[.25truecm]  
\nonumber 
&\Downarrow& \hspace{0truecm}
\text{ bounded commutator for } M_8(\bb C)\\ [.25truecm] 
\nonumber
{\cal B}' &=&(\bb H^l_L\oplus \bb H^r_L\oplus \bb H^l_R\oplus \bb H^r_R) \oplus M_4(\bb C) \\[.25truecm] 
\nonumber 
&\Downarrow&\hspace{0truecm} \text{$1^\text{st}$-order for the
  Majorana-Dirac operator} D_M\\[.25truecm]
\nonumber
 {\cal B} &=&(\bb H^l_L\oplus \bb H^r_L\oplus \bb C^l_R\oplus \bb C^r_R) \oplus M_3(\bb C)\oplus \mathbb C
 \text{ with } \bb C = \bb C_R^r \\[.25truecm] 
\nonumber 
&\Downarrow&\hspace{0truecm} \text{minimum of the spectral action}\\[.25truecm] 
\A_{sm} &=&\bb C \oplus \bb H \oplus M_3(\bb C)
\end{eqnarray*}
\end{framed}
\medskip

Starting with the ``not so grand algebra'' $\cal B$, one builds a
twisted spectral triple whose fluctuations generate both an extra
scalar field $\boldsymbol \sigma$ and an additional vector field
$X_\mu$. This is a Pati-Salam like model - the unitary of $\cal
B$ yields both an $SU(2)_R$ and an $SU(2)_L$, together with an extra $U(1)$ - but in a pre-geometric
phase since the Lorentz symmetry (in our case: the Euclidean $SO(n)$ symmetry) is
not explicit. The spectral action spontaneously breaks this model to
the standard model, in which the Lorentz symmetry is explicit, with the scalar and the vector fields playing
a role similar as the one of Higgs field. We thus have a dynamical model of emergent geometry.
\medskip

The idea that the scalar field $\sigma$ is associated to the
spontaneous breaking of a bigger symmetry to the standard model had been
formulated in 
\cite{Devastato:2013fk}, but, was not fully implemented, because the fluctuation of the
free Dirac operator by the grand algebra ${\cal A}_G$ yields an operator whose
square is a non-minimal Laplacian. The heat kernel expansion of
such operators is notably difficult to compute (see
\cite{Iochum:2016aa} for recent developments on that matter).  Almost
simultaneously, a similar idea has been successfully 
implemented in\cite{Chamseddine:2013fk}, where the Pati-Salam like symmetry does not come
from a bigger algebra, but follows from relaxing the first-order
condition.
It would be interesting to understand to what extend the twisted
fluctuations presented here are a particular case of those inner
fluctuation without first oder condition. More generally, the
  structure of the set of twisted fluctuations and of the associated
  twisted-gauge transformations of $\mathbb A$ needs to be
  worked out. Let us mention a possibly relevant notion of twisted
  connections, explained for instance in
  \cite{Fathizadeh:2011aa}.
\medskip

The twist $\rho$ is remarkably
  simple, and its mathematical significance should be studied more in
  details, in particular how it should be incorporated in the axioms
  of noncommutative geometry, like the orientability condition where
the commutator with the Dirac operator plays a crucial role. Also, the physical meaning of the twist is intriguing:
  the un-twisting of $\cal B$ forces the action of the algebra to be the same on the left and right components of
  spinors. In this sense the breaking of the grand algebra to the
  standard model is a sort of ``primordial'' chiral
  symmetry breaking.
\medskip

Full phenomenology and comparison with \cite{Chamseddine:2013uq}
require to take into account all fermions, not
only the right neutrino. This means to compute the spectral action of $\slashed D + D_M  +
\gamma^5\otimes D_0$. This would also allow to check that our $\boldsymbol
\sigma$ couples to the Higgs as $\sigma$ does in
\cite{Chamseddine:2012fk}. The simultaneous occurrence of both
  a scalar and a vector field offers interesting perspective for
  physics beyond the standard model. A similar phenomena
    appears in a recent
   proposal on how to  generate the field $\sig$ in NCG \cite{Farnsworth:2014aa}, where it
   comes together with an additional bosonic field with $B-L$ charge. 
\medskip

Finally, let us mention a very recent work of Chamseddine, Connes and
Mukhanov \cite{quanta-geom}
where the algebra $\A_F$ for $a=2$ is obtained without the
ad-hoc symplectic hypothesis, but from an higher degree Heisenberg
relation for the space-time coordinates. It
would be interesting to understand whether the case $a=4$ enters this
framework.
\bigskip

\emph{Since the first version of this paper, the twisting of real spectral triples
has been investigated in a systematic way in
\cite{Landi:2015aa} (see also \cite{T.-Brzezinski:2016aa} for an
alternative proposition). It has been shown that the 
twisted
first-order condition introduced in def. \ref{deftwistedfirstorder}
makes sense in full
generality. In particular, requiring that
the automorphism $\rho$ commutes with the real structure simply
amounts to twisting the commutator $[[D,a]_\rho, JbJ^{-1}]$ with the natural
image of $\rho$ in the group of automorphism of the opposite algebra
$\A^\circ$. Furthermore, imposing as an input that the
fermionic content of the theory is not touched (meaning that
both the Hilbert space and the Dirac operator remain unchanged), it has
been shown that there is no other choice for twisting an almost
commutative geometry than the exchange of the left/right components of spinors. There is however
some freedom in the doubling of the algebra, and the results of
\cite{Landi:2015aa} indicate that there exists at least one other model,
 in which the algebra $M_3(\mathbb C)$ is doubled as well. The
phenomenological consequence shall be explored in a
forthcoming paper.}

\vspace{.5truecm}
\begin{center}
\rule{5cm}{.5pt}
\end{center}
\vspace{.5truecm}

 \paragraph{Acknowledgments:}
The authors thank W. v. Suijlekom and J.-C. Wallet for having suggested
in a completely independent ways that twists could be a solution to
the unboundedness of commutators. Special thank to F. Lizzi for
launching the grand algebra project, many discussions and early
reading of the manuscript. Thanks also to F. D'Andrea for
pointing out misleading formulations in the first version of the
paper. Part of this
work was done during a stay of A.D. at the university of Nijmegen.

\setcounter{section}{0}
\bibliographystyle{amsplain}
\def\baselinestretch{.5}
\providecommand{\bysame}{\leavevmode\hbox to3em{\hrulefill}\thinspace}
\providecommand{\MR}{\relax\ifhmode\unskip\space\fi MR }
\providecommand{\MRhref}[2]{%
  \href{http://www.ams.org/mathscinet-getitem?mr=#1}{#2}
}
\providecommand{\href}[2]{#2}

 \end{document}